\documentclass[a4paper,11pt]{article}

  \usepackage{mathrsfs}
       \usepackage[mathscr]{euscript}

\newcommand{\longversion}[1]{#1}
\newcommand{\shortversion}[1]{}
\usepackage{}
\usepackage{amsfonts,amsthm,amsmath}
\usepackage{amstext}
\usepackage{graphicx,tikz}
\RequirePackage{fancyhdr}
\usepackage{xcolor}
\usepackage{boxedminipage}

\usepackage{vmargin}
\setmarginsrb{1.1in}{1.1in}{1.1in}{1.1in}{0mm}{0mm}{0mm}{7mm}

\usepackage[linesnumbered, boxed, algosection]{algorithm2e}

\newcommand{\myparagraph}[1]{\smallskip\noindent{\textbf{\sffamily #1} \ }}

\newtheorem{theorem}{Theorem}
\newtheorem{lemma}{Lemma}[section]
\newtheorem{claim}{Claim}[section]
\newtheorem{corollary}{Corollary}
\newtheorem{definition}{Definition}[section]
\newtheorem{observation}{Observation}[section]

\newtheorem{redrule}{Reduction Rule}

\newcommand{\defparproblem}[4]{
  \vspace{1mm}
\noindent\fbox{
  \begin{minipage}{0.96\textwidth}
  \begin{tabular*}{\textwidth}{@{\extracolsep{\fill}}lr} #1  & {\bf{Parameter:}} #3 \\ \end{tabular*}
  {\bf{Input:}} #2  \\
  {\bf{Question:}} #4
  \end{minipage}
  }
  \vspace{1mm}
}

\usepackage{graphics}

\usepackage{pgf}
\usepackage{paralist} 
\usepackage{graphicx}
\usepackage{boxedminipage}
\usepackage{boxedminipage}
\usepackage{xcolor}
\usepackage[linesnumbered,boxed,algosection]{algorithm2e}
\usepackage{framed}

\newcommand{\FPT}{\text{\normalfont FPT}}

\newcommand{\NPH}{\textsf{NP}-hard}

\usepackage{amsmath, amssymb, latexsym}
\usepackage{enumerate}

\usepackage{float}
\usepackage{xspace}

\newcommand{\delqbackdoor}{{\sc Deletion \qhorn\ Backdoor Set Detection}}

\newcommand{\skewsymmc}{$d$-{\sc Skew-Symmetric Multicut}}
\newcommand{\oneskewsymmc}{$1$-{\sc Skew-Symmetric Multicut}}
\newcommand{\threeskewsymmc}{$3$-{\sc Skew-Symmetric Multicut}}

\newcommand{\edgebip}{{\sc Edge Bipartization}}

\newcommand{\qhorn}{\text{\normalfont{$q$-Horn}}}

 \newcommand{\twocnf}{2-\text{\normalfont{CNF}}}

\newcommand{\almsat}{{\sc Almost} $2$-{\sc SAT}$(v)$}
\newcommand{\almsatc}{{\sc Almost} $2$-{\sc SAT}}
\newcommand{\agvcfull}{{\sc Above Guarantee Vertex Cover}}
\newcommand{\agvc}{{\sc AGVC}}

\newcommand{\mwc}{{\sc  Multiway Cut}}

\newcommand{\oct}{{\sc OCT}}

\newcommand{\octfull}{{\sc Odd Cycle Transversal}}

\newcommand{\bigoh}{{\mathcal{O}}}

\newcommand{\No}{{\sc No}}
\newcommand{\Yes}{{\sc Yes}}

\newcommand{\SB}{\{\,}

\newcommand{\SE}{\,\}}
\newcommand{\Card}[1]{|#1|}

\newcommand{\D}[1]{\mbox{D$($#1$)$}}
\newcommand{\var}{\mbox{var}}
\newcommand{\lit}{\mbox{lit}}

\newcommand{\tail}{\mbox{\bf Tail}}
\newcommand{\head}{\mbox{\bf Head}}

\newcommand{\CCC}{\mathcal{C}}

\newcommand{\hy}{\hbox{-}\nobreak\hskip0pt}
\newcommand{\justminus}{-}

\title{\bf Linear Time Parameterized Algorithms via Skew-Symmetric Multicuts\thanks{Supported by Parameterized Approximation, ERC Starting Grant 306992.}}

\author{
{\large\sc {M. S. Ramanujan\thanks{Institute of Mathematical Sciences, India. \texttt{\{msramanujan|saket\}@imsc.res.in}}}}\addtocounter{footnote}{-1}
\and {\large\sc Saket Saurabh\footnotemark~\thanks{University of Bergen, Norway.}}}

\date{}
\begin{document}

\maketitle

\begin{abstract} 
A skew-symmetric graph $(D=(V,A),\sigma)$ is a directed graph $D$ with an involution $\sigma$ on the set of vertices and arcs. Flows on skew-symmetric graphs have been used to generalize maximum flow and maximum matching problems on graphs, initially by Tutte [1967], and later by 
Goldberg and Karzanov [1994, 1995]. In this paper, we introduce a separation problem, \skewsymmc{}, where we are given a skew-symmetric graph $D$, a family of $\cal T$ of $d$-sized subsets of vertices and an integer $k$. The objective is to decide if there is a set $X\subseteq A$ of $k$ arcs such that every set $J$ in the family has a vertex 
$v$ such that $v$ and $\sigma(v)$ are in different connected components of $D'=(V,A\setminus (X\cup \sigma(X))$. 
In this paper, we give an algorithm for \skewsymmc{} which runs in time $\bigoh((4d)^{k}(m+n+\ell))$, where $m$ is the number of arcs in the graph, $n$ the number of vertices and $\ell$ the length of the family given in the input. 

This problem, apart from being independently interesting, also abstracts out and captures the main combinatorial obstacles towards 
solving numerous classical problems. Our algorithm for \skewsymmc{} paves the way for the first {\em linear time parameterized 
algorithms} for several problems. We demonstrate its utility by obtaining following linear time parameterized algorithms. 

\begin{itemize}
 \setlength{\itemsep}{-2pt}

  \item We show that {\almsatc} is a special case of {\oneskewsymmc}, resulting in an algorithm for {\almsatc} which runs in time $\bigoh(4^kk^4\ell)$\footnote{We have recently learnt that independently to our work and using very different techniques, Iwata, Oka and Yoshida have also given a $\bigoh(4^k\ell)$ parameterized algorithm for {\almsatc}.} where $k$ is the size of the solution and $\ell$ is the length of the input formula. Then, using linear time parameter preserving reductions to {\almsatc}, we obtain algorithms for  {\sc Odd Cycle Transversal} and {\sc Edge Bipartization} which run in time $\bigoh(4^kk^4(m+n))$  and $\bigoh(4^kk^5(m+n))$ respectively where $k$ is 
size of the solution, $m$ and $n$ are the number of edges and vertices respectively. This resolves an open problem posed by Reed, Smith and Vetta [Operations Research Letters, 2003] and improves upon the earlier almost linear time algorithm of Kawarabayashi and Reed [SODA, 2010]. 
  
\item We show that {\delqbackdoor} is a special case of {\threeskewsymmc}, giving us an algorithm for {\delqbackdoor} which runs in time $\bigoh(12^kk^5\ell)$ where $k$ is the size of the solution and $\ell$ is the length of the input formula. This gives the first fixed-parameter tractable 
algorithm for this problem answering a question posed in a paper by a superset of the authors [STACS, 2013]. 
 Using this result, we get an algorithm for {\sc Satisfiability} which runs in time $\bigoh(12^kk^5\ell)$ where $k$ is the size of the smallest $\qhorn$ deletion backdoor set, with $\ell$ being the length of the input formula.  

\end{itemize}


 \end{abstract}

\newpage
\section{Introduction}

A skew-symmetric graph is a digraph $D=(V,A)$ along with an involution $\sigma:V\cup A\rightarrow V\cup A$ where for every $x\in V\cup A$, $\sigma(x)\neq x$ and $\sigma(\sigma(x))=x$ and for every $x\in V$ ( $x\in A$), $\sigma(x)\in V$ (respectively $\sigma(x)\in A$).
Skew-symmetric graphs were introduced under the name of \emph{antisymmetrical digraphs} by Tutte~\cite{Tutte67} along with a notion of self-conjugate flows as a generalization of maximum flows in networks and matchings in graphs and subsequently by Zelinka~\cite{Zelinka76} and Zaslavsky~\cite{Zaslavsky91}. Goldberg and Karzanov~\cite{GoldbergK96,GoldbergK04} revisited the work of Tutte and gave unified proofs for the analogues of the flow-decomposition and max-flow min-cut theorems on these graphs.


In this paper we use skew-symmetric graphs and an appropriate notion of separators on them as a model to abstract out ``cut properties'' underlying 
several problems in parameterized complexity. In parameterized complexity each problem instance comes with a parameter $k$ and  a central notion in parameterized complexity is {\em fixed parameter tractability} ({\FPT}). This means, for a given instance $(x,k)$, solvability in time $f(k)\cdot p(|x|)$, where $f$ is an arbitrary function of $k$ and $p$ is a polynomial in the input size.

We now  introduce the main problem studied in this paper -- a variant of the {\sc Multicut} problem on skew-symmetric graphs.

\medskip 

 \defparproblem{{\skewsymmc}}{A skew-symmetric graph $D=((V,A),\sigma)$, a family ${\cal T}$ of $d$-sets of vertices, integer $k$.}{$k$}{Is there a set $S\subseteq A$ such that $S=\sigma(S)$, $\vert S\vert\leq 2k$, and for any $d$-set $\{v_1,\dots,v_d\}$ in $\cal T$, there is a vertex $v_i$ such $v_i$ and $\sigma(v_i)$ lie in distinct strongly connected components of $D\setminus S$?} \vspace{10 pt}
 
\medskip 
 

\noindent 
The set $S$ is the above definition is called a \emph{skew-symmetric multicut} for the given instance.
Our main result is a \FPT{} algorithm for the above problem where the dependence of the running time of the algorithm on the input size is linear. Formally,
 
 \begin{theorem}\label{thm:skewsymmc}
  There is an algorithm that, given an instance $(D=(V,A,\sigma),{\cal T},k)$ of
  {\skewsymmc}, runs in time $\bigoh((4d)^{k}k^4(\ell+m+n))$ and either returns a skew-symmetric multicut of size at
  most $2k$ or correctly concludes that no such set exists, where $m=\vert A\vert$ and $n=\vert V\vert$, and $\ell$, the length of the family ${\cal T}$, is defined as $d\cdot\vert{\cal T}\vert$.
\end{theorem}

\myparagraph{Overview of our algorithm.} The main obstacle to applying existing digraph algorithms on skew-symmetric graphs comes from the fact that standard arguments heavily based on sub-modularity of cuts break down, allowing only approximations, (see for example~\cite{MarxR09,GaspersORSS13}).
Our first contribution is a reduction rule which overcomes this obstacle by allowing us to essentially (and correctly) think of \emph{local} parts of an irreducible instance as a normal digraph. The reduction rule is essentially the following.

\begin{quote}
``{Given two vertex sets $X$ and $Y$  which satisfy certain properties, if there is a minimum $X$-$Y$ separator which contains an arc $a$ and its image $\sigma(a)$, then some arc ({\em not necessarily $a$}) in this separator  is part of an optimal solution and therefore, we find this arc, delete it from the instance, reduce the budget and continue.}''

\end{quote}

\noindent
Though simple to state, the soundness of this reduction rule is far from obvious, requires certain structural observations specific to separators in skew-symmetric graphs. Observe that this is a \emph{parameter decreasing rule} which ensures that the number of applications of this rule is bounded linearly in the parameter. We believe that this rule could prove to be of independent interest for data reduction (or kernelization) for this as well as other similar problems.

Given this reduction rule, the next obstacle we need to overcome is that of applying this rule in linear time. For this, we start from the Ford-Fulkerson algorithm for computing maximum flows and by coupling it carefully with structural properties of skew-symmetric graphs, show that in linear time, we can either apply the reduction rule or locate a part of the graph which is already ``reduced'' for the next step of our algorithm. 
In the final step of our algorithm, having found a reduced part of the graph, we show that it is sufficient for us to consider arcs emanating from this part whose number is linearly bounded in the parameter and move ahead by performing an exhaustive branching on this bounded set of arcs. Finally, by a combination of these three subroutines, we obtain a linear time {\FPT} algorithm for {\skewsymmc}. 

An additional feature of the algorithm we present is that it does not require the family ${\cal T}$ to be explicitly given as part of the input. It is sufficient for our algorithm to have access to a linear time violation oracle for ${\cal T}$, i.e, an algorithm which, in linear time returns a violated set in ${\cal T}$. This feature increases the utility of this algorithm substantially and we demonstrate this in the case of {\delqbackdoor} where even though a direct reduction does not run in linear time, we can use a linear time violation oracle to obtain a linear time {\FPT} algorithm for {\delqbackdoor}.

\medskip

\myparagraph{Applications.} Here we provide the list of main applications (see Figure~\ref{fig:zoo}) that can be derived from our
algorithm (Theorem~\ref{thm:skewsymmc}) together with a short overview of previous work on each application.

%
%
%
%

\myparagraph{{\octfull}, {\almsatc} and related problems.} Graph bipartization is a classical {\NPH} problem with several applications~\cite{choi1989graph,rizzi2002practical,panconesi2004fast}. In the field of parameterized complexity, this problem is better known as {\octfull}, which is formally defined as follows.

\medskip

\defparproblem{\octfull({\oct})}{A graph $G$, positive integer $k$}{$k$}{Does there exist a set $S$ of at most $k$ vertices such that $G\setminus S$ is a bipartite graph?}\vspace{10 pt}

\medskip 

The parameterized complexity of {\sc Odd Cycle Transversal} was a well known open problem for a long time. In 2003,  in a breakthrough paper,  Reed et al.~\cite{ReedSV04}  showed that {\oct}  is {\FPT} by developing an algorithm for the problem running in time $\bigoh(3^kmn)$.  
We use $n$ and $m$ to denote the number of vertices and edges of the input graph respectively.  
In fact this was the first time that the iterative compression technique was used. This technique has been useful in resolving several other open problems in the area of parameterized complexity, including {\sc Directed Feedback Vertex Set}, {\almsatc}, {\sc Multicut}~\cite{ChenLLOR08,RazSul09,MarxRazmulticut}. However, the algorithm for {\sc OCT} had seen no further improvements in the last $9$ years, though reinterpretations of the algorithm have been published~\cite{Huffner09,LokshtanovSS09}. Only recently,  Lokshtanov et al.~\cite{LNRRS12} obtained an algorithm with an improved dependence on the parameter $k$. This algorithm is based on a branching guided by linear programming and runs in time $\bigoh(2.32^kn^{\bigoh(1)})$.
In a parallel line of research, Fiorini et al.~\cite{FioriniHRV08} showed that when the input is restricted to planar graphs there is an $\bigoh(f(k)n)$ time algorithm-- a linear time algorithm, for {\oct}. Continuing this line of research, recently, Kawarabayashi and Reed~\cite{KawarabayashiR10} obtained an algorithm for {\oct} on general graphs with an improved dependence on the input size. This algorithm uses tools from graph minors and odd variants of graph minors and runs in time $\bigoh(f(k)m \cdot \alpha(m,n))$. Here the function $\alpha(m,n)$ is the inverse of the Ackermann function (see by Tarjan~\cite{Tarjan75}) and $f(k)$ is at least a triple exponential function. However, an algorithm on general graphs with a linear dependence on the input size has so far proved elusive.
In this work, we obtain the first linear time algorithm for {\oct} running in time $\bigoh(4^kk^4(m+n))$. This resolves an open problem posed by Reed et al. in 2003~\cite{ReedSV04}.

In the \emph{edge} version of {\oct}, namely {\sc Edge Bipartization}, the objective is to test if there is a set of at most $k$ edges whose deletion makes the input graph bipartite. Using a known linear time parameter preserving reduction from {\edgebip} to {\oct}, we also get a similar result for {\edgebip}. 
In fact, both these problems have linear time parameter preserving reductions to the more general problem of {\almsatc}. In fact, our algorithms for {\edgebip} and {\oct} are obtained via reductions to {\almsatc}, which in turn is solved using our algorithm for {\skewsymmc} (Theorem~\ref{thm:skewsymmc}).

 \begin{figure}[t] \label{fig:zoo}
\begin{center}
\includegraphics[height=300 pt, width=300 pt]{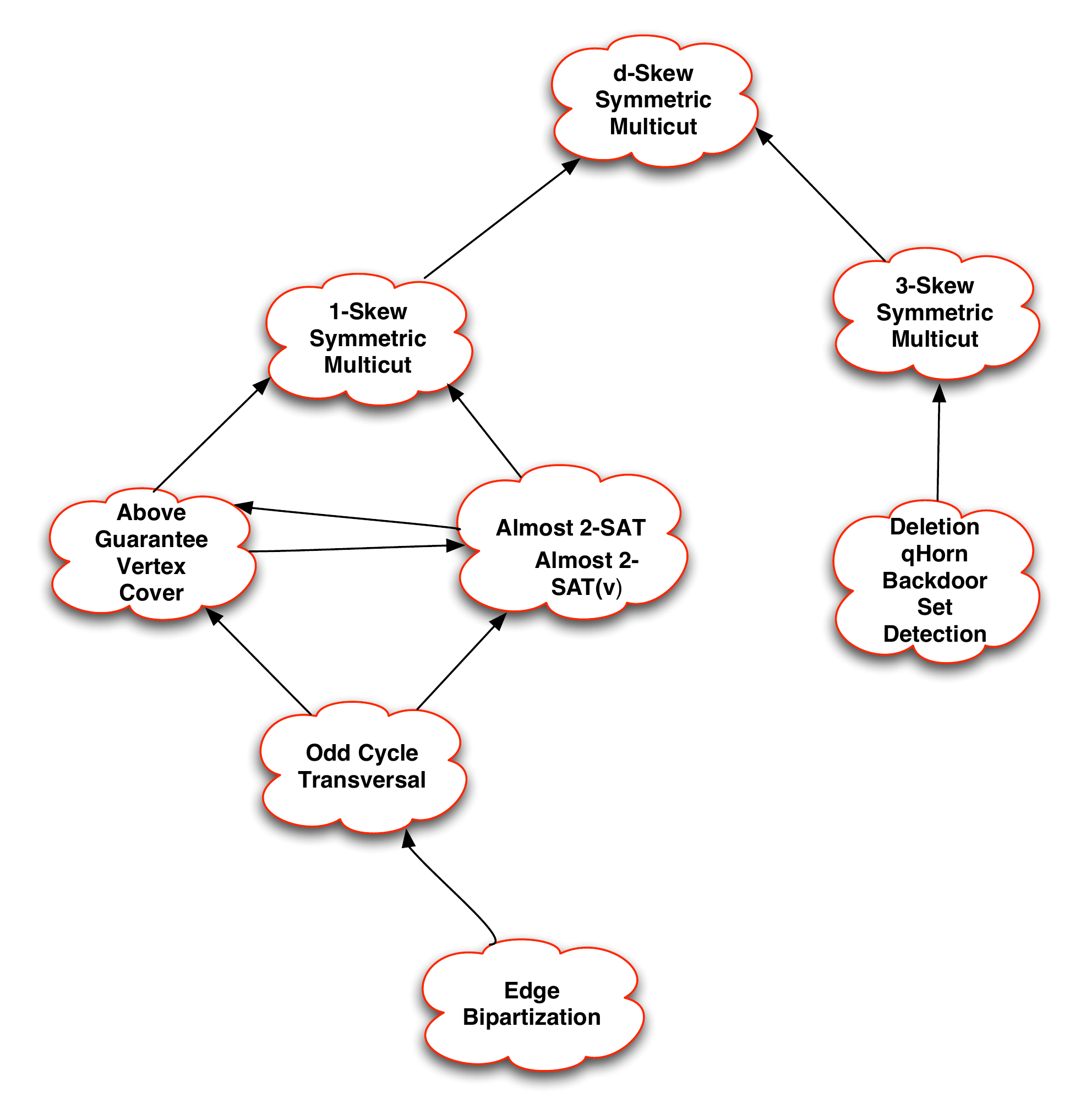}
 \caption{Problems with a linear time parameter preserving reduction to {\skewsymmc}}
\end{center}
\end{figure}


The {\almsatc} problem is formally defined as follows.
\medskip 

\defparproblem{\almsatc}{A \twocnf{} formula $F$, positive integer $k$}{$k$}{Does there exist a set $S_c$ of at most $k$ clauses of $F$ such that $F\setminus S_c$ is satisfiable? }\vspace{10 pt}


\noindent
It was introduced by Mahajan and Raman~\cite{MahajanRaman99} in 1999 and its parameterized complexity status remained open until 2008 when Razgon and O'Sullivan~\cite{RazgonOSullivan08} gave an algorithm running in time $\bigoh(15^kkm^3)$ on formulas with $m$ clauses. 
More recently, there have been a series of improved algorithms ($\bigoh(9^kn^{\bigoh(1)}$)\cite{RamanRS11}, $\bigoh(4^kn^{\bigoh(1)}$)\cite{CyganPPW11}, $\bigoh(2.618^kn^{\bigoh(1)}$)\cite{NarayanaswamyRRS12}) with the current best algorithm running in time $\bigoh(2.32^kn^{\bigoh(1)})$~\cite{LNRRS12}. However, none of these algorithms have linear dependence on the input size.
We show that {\almsatc} is a special case of {\oneskewsymmc}, resulting in an algorithm for {\almsatc} which runs in time $\bigoh(4^kk^4\ell)$ where $k$ is the size of the solution and $\ell$ is the length of the input formula.

Another problem related to {\almsatc} is the {\sc Above Guarantee Vertex Cover} ({\sc agvc}) which is defined as follows.

\medskip

\defparproblem{\agvcfull\  (\agvc)}{$(G=(V,E),M,k)$, where $G$ is an undirected graph, $M$ is a maximum matching for $G$, $k$ a positive integer}{$k-\vert M\vert$}{Does $G$ have a vertex cover of size at most $k$?} 

\medskip
\noindent
This problem is linear time equivalent to {\almsatc} in a parameter preserving way and hence, our results imply an algorithm for {\agvc} with running time $\bigoh(4^{k-\vert M\vert}(k-\vert M\vert)^4(m+n))$. This equivalence has already proved useful as a linear programming based branching algorithm for {\agvc} led to algorithms for {\almsatc} and several other problems around it. However, the linear time reduction from {\agvc} to {\almsatc} crucially utilizes the fact that a maximum matching of the graph is also part of the input. Therefore, a natural goal would to obtain an algorithm running in time $f(k-\vert M\vert)(m+n)$ for {\agvc} even when a maximum matching is not given as part of the input. 

\medskip 

\myparagraph{{\delqbackdoor} and related problems.} A \emph{strong} \emph{$\CCC$\hy backdoor set} of a CNF formula $F$ is a
set $B$ of variables such that $F[\tau]\in \CCC$ for each assignment
$\tau : B \rightarrow \{0,1\}$, that is, for every instantiation of the variables in $B$, the reduced formula is in the class $\CCC$. A \emph{deletion $\CCC$\hy backdoor set} of $F$
is a set $B$ of variables such that $F \justminus B \in \CCC$.  Backdoor
sets were independently introduced by Crama et
al.~\cite{CramaEkinHammer97} and by Williams et al.~\cite{WilliamsGomesSelman03}, the latter authors coining the term ``backdoor''.
If we know a strong $\CCC$\hy backdoor set of $F$ of size $k$, we can
reduce the satisfiability of $F$ to the satisfiability of $2^k$ formulas
in $\CCC$. If $\CCC$ is
clause-induced, every deletion $\CCC$\hy backdoor set of $F$ is a strong
$\CCC$\hy backdoor set of $F$. For several base classes, deletion
backdoor sets are of interest because they are easier to detect than
strong backdoor sets. 

The parameterized complexity of finding small
backdoor sets was initiated by Nishimura et
al.~\cite{NishimuraRagdeSzeider04-informal} who showed that for the base
classes of Horn formulas and Krom formulas, the detection of strong
backdoor sets is fixed-parameter tractable. For base classes other than Horn and Krom,
strong backdoor sets can be much smaller than deletion backdoor sets,
and their detection is more difficult. 
For more recent results, the reader is
referred to a  survey on the parameterized complexity of backdoor
sets~\cite{GaspersSzeider12}.

The class $\qhorn$, introduced by Boros, Crama and Hammer~\cite{BorosHammerSun94}, is
  one of the largest known classes of propositional CNF formulas for
  which satisfiability can be decided in polynomial time. This class properly
  contains the fundamental classes of Horn and Krom formulas as well as
  the class of renamable (or disguised) Horn formulas. The parameterized complexity of finding small $\qhorn$ backdoor sets was studied by Gaspers et al.~\cite{GaspersORSS13} who showed that the {\delqbackdoor} problem is fixed-parameter approximable. Formally, the {\delqbackdoor}  problem is the following.

\medskip

\defparproblem{{\delqbackdoor}}{
  A CNF formula $F$ and a positive integer $k$}{$k$}{Does $F$ have a
  deletion $\qhorn$ backdoor set of size at most $k$?}

\medskip

We show that {\delqbackdoor} is a special case of {\threeskewsymmc}, giving us an algorithm for {\delqbackdoor} which runs in time $\bigoh(12^kk^5\ell)$ where $k$ is the size of the solution and $\ell$ is the length of the input formula. This gives the first fixed-parameter tractable 
algorithm for this problem.  Using this result, we get an algorithm for {\sc Satisfiability} which runs in time $\bigoh(12^kk^5\ell)$ where $k$ is the size of the smallest $\qhorn$ deletion backdoor set, with $\ell$ being the length of the input formula.

\myparagraph{Related Results on Cut problems.}  {\skewsymmc} is neither the first nor will it be the last cut problem studied in parameterized complexity. 
Marx~\cite{Marx06} was the first to consider graph separation problems 
in the context of parameterized complexity. He gave an algorithm for \mwc{} 
with a running time of $\bigoh(4^{k^3}n^{\bigoh(1)})$ which was later improved to $\bigoh(4^{k}n^{\bigoh(1)})$ by Chen at al.~\cite{ChenLL09}. The current fastest algorithm for \mwc{} runs in time $\bigoh(2^{k}n^{\bigoh(1)})$~\cite{CPPO11}. The notions 
used in the paper of Marx have been useful in settling the parameterized complexity, as well as obtaining improved {\FPT} algorithms for a wide variety of problems including 
{\sc Directed Feedback vertex Set}~\cite{ChenLLOR08}, {\almsatc}~\cite{RazSul09} and {\sc Above Guarantee Vertex Cover}~\cite{RazSul09,RamanRS11}. These sequence of results have led to an entirely new and very active subarea dealing with parameterized graph separation problems both due to independent interest in the problems themselves, as well as due to the fact that these problems seem to be able to capture the underlying properties of a large variety of seemingly unrelated problems.

\myparagraph{Organization of the paper.} In Section~\ref{sec:skewsym properties}, we define the notions of separators in skew-symmetric graphs, followed by structural results on separators in skew-symmetric graphs and the notions of $(L,k)$-components whose computation is at the core of our algorithm. In Section~\ref{sec:algo}, we prove an observation regarding the structure of optimal solutions, followed by a description and proof of correctness of our algorithm. In Section~\ref{sec:apps}, we give linear time parameterized algorithms for a number of problems using our result.

\section{Preliminaries}\label{sec:prelims}
In this section we give some basic definitions and set up the notations for the paper. 

\myparagraph{Parameterized Complexity.} Parameterized complexity is one of the ways to cope up with intractability of problems. 
The goal of parameterized complexity is to find ways of solving {\NPH} problems more efficiently than by brute force. 
Here, the aim is to restrict the combinatorial explosion of computational difficulty to a parameter that is hopefully much smaller than the input size.
Formally, a {\em parameterization} of a problem is the assignment of an integer $k$ to each input instance and  we say that a parameterized problem is {\em fixed-parameter tractable} ({\FPT}) if there is an algorithm that solves the problem in time $f(k)\cdot |I|^{\bigoh(1)}$, where $|I|$ is the size of the input instance and $f$ is an
arbitrary computable function depending only on the parameter $k$. 
For more background, the reader is referred to the monographs \cite{DF99,FG06,Nie06}.

\medskip 
\myparagraph{Digraphs.} Let $D=(V,A)$ be a directed graph. For an arc $(u,v)\in A$, we refer to $u$ as the \emph{tail} of this arc and denote it by $\tail(u,v)$ and we refer to $v$ as the \emph{head} of this arc and denote it by $\head(u,v)$. For a set of arcs $P$, we denote by $\tail(P)$, the set $\bigcup_{(u,v)\in P} \{\tail(u,v)\}$ and we denote by $\head(P)$ the set $\bigcup_{(u,v)\in P} \{\head(u,v)\}$. For a set of vertices $V'$, we let $A[V']$ denote the set of arcs with both end points in the set $V'$. For a set of vertices $V'$, we let $\delta^+(V')$ denote the set of arcs which have their tail in $V'$ and their head in $V\setminus V'$. Similarly, we let $\delta^-(V')$ denote the set of arcs which have their head in $V'$ and their tail in $V\setminus V'$. We also use $N^+(V')$ to denote the set $\head (\delta^+(V'))$ and $N^-(V')$ to denote the set $\tail (\delta^-(V'))$. Given two disjoint vertex sets $X$ and $Y$, we define an $X$-$Y$ path as a directed path from a vertex $x\in X$ to a vertex $y\in Y$ whose internal vertices are disjoint from $X\cup Y$.

\medskip 
\myparagraph{Skew-Symmetric Graphs.} The notation is from \cite{GoldbergK04}. A \emph{skew-symmetric graph} is a digraph $D=(V,A)$ and an involution $\sigma:V\cup A\rightarrow V\cup A$ such that:

\begin{enumerate}
 \item for each $x\in V\cup A$, $\sigma(x)\neq x$ and $\sigma(\sigma(x))=s$.
 \item for each $v\in V$, $\sigma(v)\in V$
 \item for each $a=(v,w)\in A$, $\sigma(a)=(\sigma(w),\sigma(v))$
\end{enumerate}

We call $\sigma(x)$ \emph{symmetric} to $x$ and also refer to $x$ and $\sigma(x)$ as \emph{conjugates}.
For ease of description, we let $x'$ denote the conjugate of an element $x$ and we let $S'$ denote the set of conjugates of the elements in the set $S$. We say that a set $S$ is \emph{regular} if $S\cap S'=\emptyset$ and \emph{irregular} otherwise. A set $S$ is called \emph{self-conjugate} if $S=S'$.

\medskip 

\myparagraph{CNF Formulas and Satifiability.}  A \emph{literal} is a variable $x$ or a negated
variable $\bar x$; if $y=x$ or $y=\bar x$ is a literal for some variable
$x$, then we write $\bar y$ to denote $\bar x$ or $x$, respectively.
A \emph{clause} is a
finite set of literals and  a finite set of clauses is a \emph{CNF formula} where the clauses are considered as a disjunction of its literals and the CNF formula is considered a conjunction of its clauses.  By $\mathscr{C}(F)$ we denote the set of clauses of a CNF formula $F$. 
A formula is \emph{Horn} if each of its clauses
contains at most one positive literal, a formula is \emph{Krom} (or
\emph{2CNF}, or \emph{quadratic}) if each clause contains at most two
literals.  
The \emph{length} of a CNF
formula $F$
is defined as $\sum_{C\in F}
\Card{C}$.  If $F$ is a formula and $X$ a set of variables, then we
denote by $F\justminus X$ the formula obtained from $F$ after removing
all literals with a variable in $X$ from the clauses in $F$. 
Let $F$ be a formula and $X \subseteq \var(F)$.
A \emph{truth assignment} is a mapping $\tau:X\rightarrow \SB 0,1 \SE$
defined on some set $X$ of variables.
A truth assignment $\tau$ \emph{satisfies} a clause $C$ if $C$ contains some literal $x$ with $\tau(x)=1$; $\tau$ satisfies a formula $F$ if it satisfies all
clauses of $F$.  
A formula is
\emph{satisfiable} if it is satisfied by some truth assignment; 
otherwise it is \emph{unsatisfiable}.

\section{Skew-symmetric graphs, separators and components}\label{sec:skewsym properties}

The following observation is a direct consequence of the definition
of a skew-symmetric graph.

\begin{observation}\label{obs:implication graphs} Let $D=((V,A),\sigma)$ be a skew-symmetric graph and let $u,v\in V$. There is a path from $v$ to $u$ in $D$ iff there is a path from $u'$ to $v'$.
\end{observation}

\begin{definition} Let $D=(V,A)$ be a directed graph and let $X,Y$ be disjoint subsets of $V$. A set $S\subseteq A$ is an $X$-$Y$ separator if there is no directed path from $X$ to $Y$ in the graph $D\setminus S$. We say that $S$ is a minimal $X$-$Y$ separator if no proper subset of $S$ is an $X$-$Y$ separator. 
\end{definition}

\begin{definition}
Let $D=((V,A),\sigma)$ be a skew-symmetric graph and let $L$ be a regular set of vertices.
Let $X\subseteq A$ be a self-conjugate set of arcs of $D$.
We call $X$ an $L$-$L'$ {\bf self-conjugate separator} if $X$ is a (not necessarily minimal) $L$-$L'$  separator. We call $X$ a minimal $L$-$L'$ self-conjugate separator if there is no self-conjugate strict subset of $X$ which is also an $L$-$L'$ separator. 
\end{definition}

\begin{definition}
Let $D=((V,A),\sigma)$ be a skew-symmetric graph and let $L$ be a regular set of vertices. Let $X$ be an $L$-$L'$ self-conjugate separator. We denote by $R(L,X)$ the set of vertices of $D$ that can be reached from $L$ 
via directed paths 
in $D\setminus X$, and we denote by $\bar R(L,X)$ the set of vertices of $D$ which 
have a directed path to 
can reach $L$ in $D\setminus X$.
\end{definition}

\begin{observation}\label{obs:restate}
Let $D=((V,A),\sigma)$ be a skew-symmetric graph and let $L$ be a regular set of vertices. Let $X$ be an $L$-$L'$ self-conjugate separator. Then, the sets $R(L,X)$ and $\bar R(L',X)$ are also regular and $\sigma(R(L,X))=\bar R(L',X)$.
\end{observation}

\begin{proof} Since deleting a self-conjugate set of arcs from a skew-symmetric graph results in a skew-symmetric graph, we know that there is a path from $u$ to $v$ in $D\setminus X$ iff there is a path from $v'$ to $u'$ in $D\setminus X$. Therefore, if $R(L,X)$ is irregular, then there is a path from $L$ to $y$ and $y'$ for some vertex $y$, which is disjoint from $X$, which implies a path from $L$ to $L'$ in $D\setminus X$, which is a contradiction. Therefore, $R(L,X)$ and $\bar R(L',X)$ are regular and since $D\setminus X$ is a skew-symmetric graph, they are conjugates.
\end{proof}

\subsection{Minimum separators in skew-symmetric graphs}

\begin{lemma}\label{lem:lk component}
Let $D=((V,A),\sigma)$ be a skew-symmetric graph, $L$ be a regular set of vertices.
\begin{enumerate}
 \item Suppose that there is an $L$-$L'$ path in $D$ and let $X$ be a minimum $L$- $L'$ separator and let $Z=R(L,X\cup X')$. Then, $\delta^+(Z)$ is also a minimum $L$-$L'$ separator. 
\item An arc is part of a minimum $L$-$L'$ separator if and only if its conjugate is also part of a minimum $L$-$L'$ separator.
\end{enumerate}
\end{lemma}

\begin{proof}
 Recall that $Z$ is regular (Observation~\ref{obs:restate}). Since $L$ is in $Z$ and $L'$ is disjoint from $Z$, $\delta^+(Z)$ is an $L$-$L'$ separator. It remains to show that it is a minimum such separator. Clearly, $\delta^+(Z)\subseteq X\cup X'$. Let $A=\delta^+(Z)\cap X$ and $B=\delta^+(Z)\setminus X$. 

 Since $B$ is disjoint from $X$, it must be the case that $B'\subseteq X$. We now claim that $A$ and $B'$ are disjoint. Suppose that this is not the case and let $x\in B$ such that $x'\in A$. Since $x'\in \delta^+(Z)$, it must be the case that $x\in \delta^-(Z')$. Since there is a path from $Z$ to $Z'$ via $x$ and $X$ is disjoint from $A[Z\cup Z']\cup \{x\}$, there is a path from $L$ to $L'$ disjoint from $X$, a contradiction. Therefore, we conclude that $A\cap B'=\emptyset$. We now have that $\vert \delta^+(Z)\vert=\vert A \cup B\vert=\vert A\vert +\vert B'\vert\leq \vert X\vert$ where the last inequality used the fact that $A$ and $B'$ are disjoint. Therefore, we conclude that $\delta^+(Z)$ is indeed a minimum $L$-$L'$ separator. Consequently, $\delta^-(Z)$ is also a minimum $L$-$L'$ separator. This concludes the proof of the first part of the lemma. 
 
We claim that if $X$ is an $L$-$L'$ separator, then $X'$ is an $L$-$L'$ separator as well. Suppose that this is not the case and let there be a path $v_1,\dots,v_r$ in $D\setminus X'$ where $v_1=l$ and $v_r=l'$. However, this implies that $X$ is disjoint from the path $\sigma(v_r),\dots,\sigma(v_1)$, which is a contradiction.
 This concludes the proof of the lemma. 
 \end{proof}
 
 \begin{lemma}\label{lem:crossing}(Crossing-Uncrossing)
 Let $D=((V,A),\sigma)$ be a skew-symmetric graph and let $B\subseteq V$. If $\delta^+(B)\cap \delta^+(B')=\emptyset$ then $\vert\delta^+(B\setminus B')\vert=\vert\delta^+(B)\vert$ and $\vert\delta^+(B\setminus B')\vert<\vert\delta^+(B)\vert$ otherwise.
\end{lemma}

\begin{proof}

 Let $Q=B\setminus B'$. We partition $\delta^+(Q)$ into the following sets (see Figure~\ref{fig:crossing}).

%
%

\begin{enumerate}
\setlength{\itemsep}{-2pt}
 \item $Q_1^{o}=\delta^+(Q)\cap \delta^-(V\setminus (B\cup B'))$, that is, those arcs with the tail in $Q$ and the head in $V\setminus (B\cup B')$.
 
 \item $Q_2^{o}=\delta^+(Q)\cap \delta^-(B\setminus B')$, that is, those arcs with the tail in $Q$ and the head in $B'\setminus B$.
 
 \item $Q_3^{o}=\delta^+(Q)\cap \delta^-(B\cap B')$), that is, those arcs with the tail in $Q$ and the head in $B\cap B'$.
\end{enumerate}

Similarly, we partition $\delta^+(B)$ as follows.

\begin{enumerate}
\setlength{\itemsep}{-2pt}
 \item $B_1^{o}=(\delta^+(B\setminus B')\cap \delta^-(V\setminus (B\cup B'))$, that is, those arcs with the tail in $B\setminus B'$ and the head in $V\setminus (B\cup B')$. 
 
 \item $B_2^{o}=(\delta^+(B)\cap \delta^-(B'\setminus B)$, that is, those arcs with the tail in $B\setminus B'$ and the head in $B'\setminus B$.
 
 \item $B_3^{o}=\delta^+(B)\cap \delta^+(B'\setminus B)$, that is, those arcs with the tail in $B\cap B'$ and the head in $B\setminus B'$.
 
 \item $B_4^{o}=\delta^+(B)\cap \delta^+(B')$ 
\end{enumerate}

\noindent
Observe that $Q_1^{o}=B_1^o$, $Q_2^{o}=B_2^o$ and $Q_3^{o}= (B_3^o)'$. Therefore, $\vert\delta^+(Q)\vert=\vert\delta^+(B)\vert$ if $B_4^o$ is empty and $\vert\delta^+(Q)\vert<\vert\delta^+(B)\vert$ otherwise. This completes the proof of the lemma. 
\end{proof}

\begin{figure}[t] \begin{center}
\includegraphics[height=200 pt, width=230 pt]{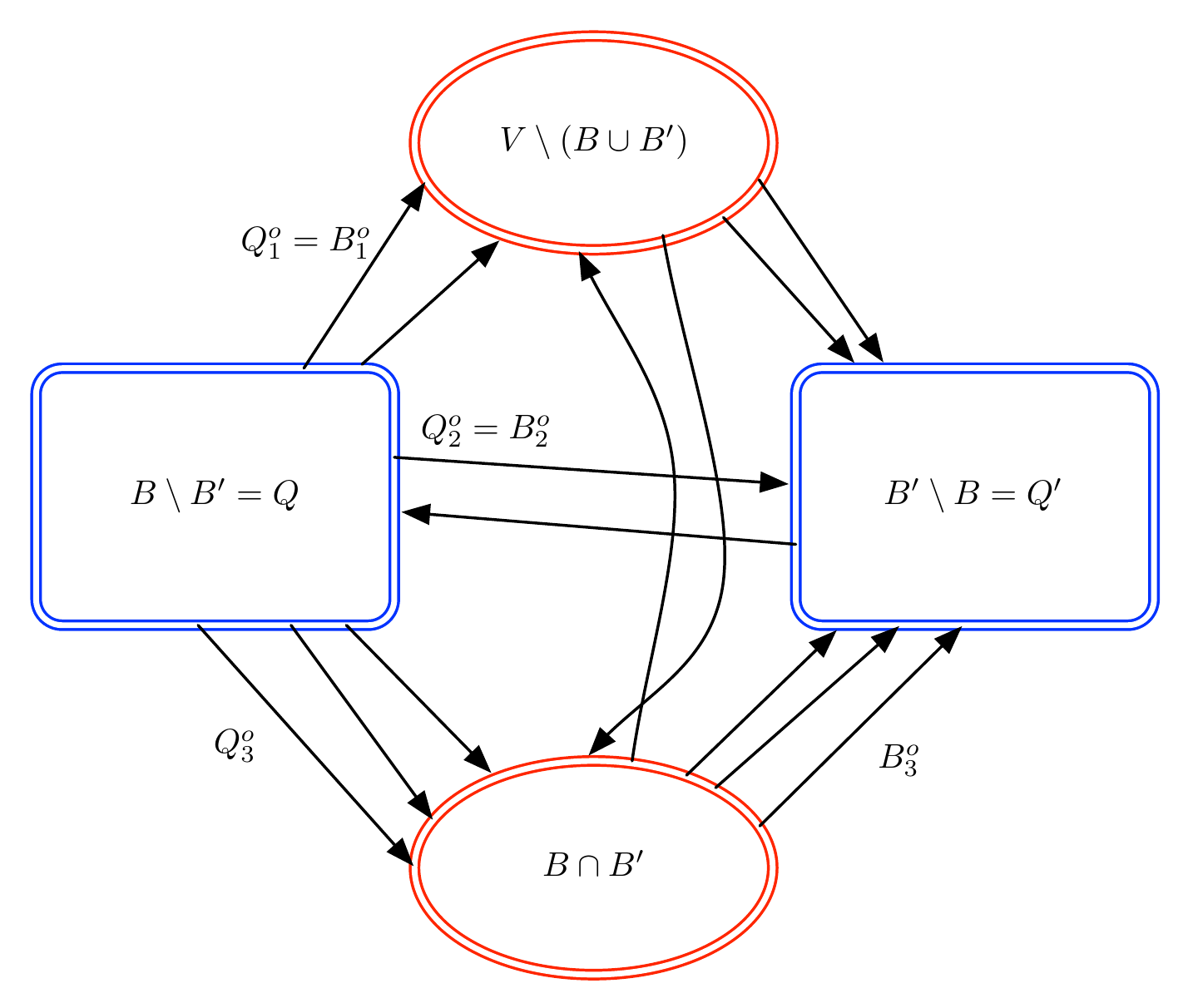}

 \caption{An illustration of the partitions described in the proof of Lemma~\ref{lem:crossing}}
\label{fig:crossing}
\end{center}
\end{figure}

\subsection{$(L,k)$-Components}

\begin{definition}
Let $D=((V,A),\sigma)$ be a skew-symmetric graph and $k\in \mathbb{N}$. Let $L\subseteq V$ be a regular set of vertices. A set of vertices $Z\subseteq V$ is called an $(L,k)$-component if it satisfies the following properties.

\begin{enumerate}
\setlength{\itemsep}{-2pt}
\item $L\subseteq Z$
\item $Z$ is regular
\item $Z$ is reachable from $L$ in $D[Z]$
\item The size of a minimum $Z$-$Z'$ separator is equal to the size of a minimum $L$-$L'$ separator and this size is at most $2k$.
\item $Z$ is inclusion-wise maximal among the sets satisfying the above properties.
\end{enumerate}
\end{definition}

\noindent
Lemma~\ref{lem:find imp set skewsym} gives an algorithm that in linear time either computes an $(L,k)$-component or finds a minimum $L$-$L'$ separator of a particular kind, which we later show can be used to reduce the instance. We first require the arc version of Lemma 2.4 from \cite{MarxOR10}.

\begin{lemma}\label{lem:separator layer}
 Let $s,t$ be two vertices in a digraph $D=(V,A)$ such that the minimum size of an $s$-$t$ separator is $\ell>0$. Then, there is a collection ${\cal X}=\{X_1,\dots,X_q\}$ of sets where $\{s\}\subseteq X_i\subseteq V\setminus \{t\}$ such that
 
 \begin{enumerate}
\setlength{\itemsep}{-2pt}
  \item $X_1\subset X_2\subset \cdots \subset X_q$,
  \item $X_i$ is reachable from $s$ in $D[X_i]$,
  \item $\vert \delta^+(X_i)\vert=\ell$ for every $1\leq i\leq q$ and 
  \item every $s$-$t$ separator of size $\ell$ is fully contained in $\bigcup_{i=1}^q \delta^+(X_i)$.
 \end{enumerate}
 
 \noindent
 Furthermore, there is an $\bigoh(\ell(\vert V\vert +\vert A\vert))$ time algorithm that produces the sets $X_1,X_2\setminus X_1,\dots, X_q\setminus X_{q-1}$ corresponding to such a collection $\cal X$.

\end{lemma}

\begin{proof}
This proof is  from \cite{MarxOR10}. However, the proof in \cite{MarxOR10} does not need that $X_i$ is reachable from $s$ in $D[X_i]$. Thus, we need to do a bit more extra work for the version presented here.  
 We first run $\ell$ iterations of the Ford-Fulkerson algorithm on the graph with unit capacities on all arcs to find a maximum $s$-$t$ flow. Let $D_1$ be the residual graph. Let $C_1,\dots, C_q$ be a topological ordering of the strongly connected components of $D_1$ such that $i<j$ if there is a path from $C_i$ to $C_j$. Recall that there is a $t$-$s$ path in $D_1$. Let $C_x$ and $C_y$ be the strongly connected components of $D_1$ containing $t$ and $s$ respectively. Since there is a path from $t$ to $s$ in $D_1$, $x<y$. For each $x<i\leq y$, let $Y_i=\bigcup_{j=i}^qC_j$ (see Figure~\ref{fig:marxlemmapic}). We first show that $\vert \delta^+(Y_i)\vert=\ell$. Since no arcs leave $Y_i$ in the graph $D_1$, no flow enters $Y_i$ and every arc in $\delta^+(Y_i)$ is saturated by the maximum flow. Therefore, $\vert\delta^+(Y_i)\vert=\ell$. 
 
 \begin{figure}[t] \begin{center}
\includegraphics[height=180 pt, width=340 pt]{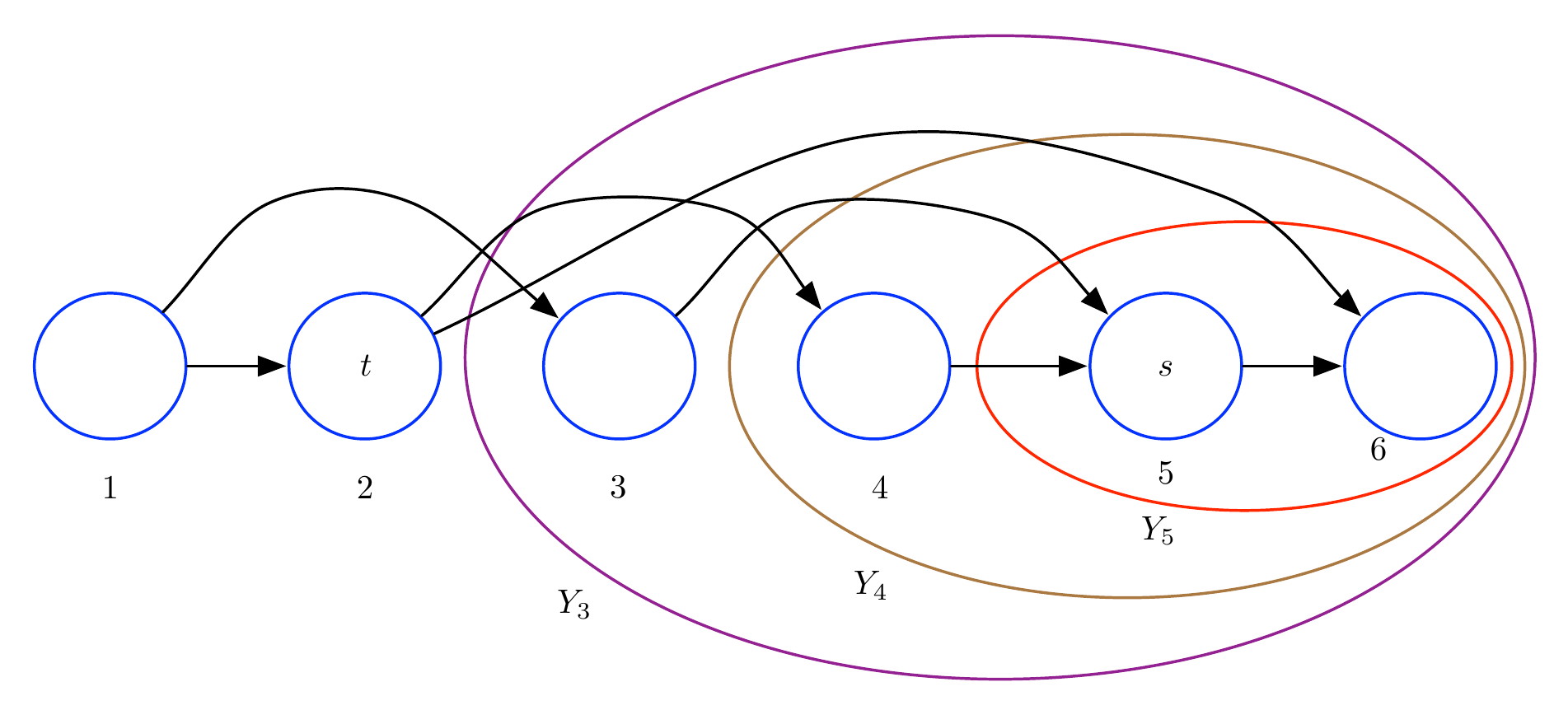}

 \caption{An illustration of the sets in the proof of Lemma~\ref{lem:separator layer}. The blue circles are the strongly connected components of $D_1$ and $\alpha(s)=5$ and $\alpha(t)=2$.}
\label{fig:marxlemmapic}
\end{center}
\end{figure}
 
 We now show that every arc which is part of a minimum $s$-$t$ separator is contained in $\bigcup_{i=1}^q \delta^+(Y_i)$. Consider a minimum $s$-$t$ separator $S$ and an arc $(a,b)\in S$. Let $Y$ be the set of vertices reachable from $s$ in $D\setminus S$. Since $F$ is a minimum $s$-$t$ separator, it must be the case that $\delta^+(Y)=F$ and therefore, $\delta^+(Y)$ is saturated by the maximum flow. Therefore, we have that $(b,a)$ is an arc in $D_1$. Since no flow enters the set $Y$, there is no cycle in $D_1$ containing the arc $(b,a)$ and therefore, if the strongly connected component containing $b$ is $C_{i_b}$ and that containing $a$ is $C_{i_a}$, then $i_b<i_a$. Furthermore, since there is flow from $s$ to $a$ from $b$ to $t$, it must be the case that $x<i_b<i_a<y$ and hence the arc $(a,b)$ appears in the set $\delta^+(Y_{i_a})$.
 
 Finally, we define the set $R(Y_i)$ to be the set of vertices of $Y_i$ which are reachable from $s$ in the graph $D[Y_i]$. Then, the sets 
 $R(Y_{y})\subset R(Y_{y-1}) \subset \cdots \subset R(Y_{x+1})$ form a collection of the kind described in the statement of the lemma.
 
In order to compute these sets, we first need to run the Ford-Fulkerson algorithm for $\ell$ iterations and perform a topological sort of the strongly connected components of $D_1$. This takes time $\bigoh(\ell(\vert V\vert+\vert A\vert))$. During this procedure, we also assign indices to the strongly connected components in the manner described above, that is, $i<j$ if $C_i$ occurs before $C_j$ in the topological ordering. In $\bigoh(\ell(\vert V\vert+\vert A\vert))$ time, we can assign indices to vertices such that the index of a vertex $v$ (denoted by $\alpha(v)$) is the index of the strongly connected component containing $v$. We then perform the following preprocessing for every vertex $v$ such that $\alpha(v)<\alpha(s)$. We go through the list of in-neighbors of $v$ and find 
$$\beta(v)=\max_{u\in N^-(v)} \Big\{ \alpha(u)~\vert~ \alpha(v)< \alpha(u) \Big\} \mbox{ and }$$  
$$\gamma(v)=\min_{u\in N^-(v)} \Big\{ \alpha(u)~\vert~ \alpha(v)< \alpha(u)\Big\}$$ 
and set $\beta'(v)=\min\{\beta(v),\alpha(s)\}$ and $\gamma'(v)=\max\{\gamma(v),\alpha(t)+1\}$.

 The meaning of these numbers is that the vertex $v$ occurs in each of the sets $N^+(Y_{\beta'(v)})$,  $N^+(Y_{\beta'(v)-1})$, $\dots$, $N^+(Y_{\gamma'(v)})$. 
  This preprocessing can be done in time $\bigoh(m+n)$ since we only compute the maximum and minimum in the adjacency list of each vertex. 
 A vertex $v$ is said to be $i$-forbidden for all $\gamma'(v)\leq i\leq \beta'(v)$. We now describe the algorithm to compute the sets in the collection.

\medskip

\noindent
{\bf Computing the collection.} We do a modified (directed) breadth first search (BFS) starting from $s$ by using only \emph{out-going} arcs. Along with the standard BFS queue, we also maintain an additional \emph{forbidden} queue.

We begin by setting $i=\alpha(s)$ and start the BFS by considering the out-neighbors of $s$.
We add a vertex to the BFS queue only if it is both unvisited and not $i$-forbidden. 
If a vertex is found to be $i$-forbidden (and it is not already in the forbidden queue), we add this vertex to the forbidden queue. Finally, when the BFS queue is empty and every unvisited out-neighbor of every vertex in this tree is in the forbidden queue, we return the set of vertices \emph{added} to the BFS tree in the current iteration as $R(Y_{i})\setminus R(Y_{i+1}) $.   Following this, the vertices in the forbidden queue which are not $(i-1)$-forbidden are removed and added to the BFS queue and the algorithm continues after decreasing $i$ by 1. 
The algorithm finally stops when $i=\alpha(t)$.

We claim that this algorithm returns each of the sets $R(Y_{\alpha(s)})$, $R(Y_{\alpha(s)-1})\setminus R(Y_{\alpha(s)})$, $\dots,R(Y_{\alpha(t)+1})\setminus R(Y_{\alpha(t)+2})$ and runs in time $\bigoh(\ell(\vert V\vert +\vert A\vert)$. In order to bound the running time, first observe that the vertices which are $i$-forbidden are exactly the vertices in the set 
$N^+(R(Y_i))$
 and therefore the number of $i$-forbidden vertices for each $i$ is at most $\ell$. This implies that the number of vertices in the forbidden queue at any time is at most $\ell$. Hence, testing if a vertex is $i$-forbidden or already in the forbidden queue for a fixed $i$ can be done in time $\bigoh (\ell)$. Therefore, the time taken by the algorithm is $\bigoh(\ell)$ times the time required for a BFS in $D$, which implies a bound of $\bigoh(\ell(\vert V\vert +\vert A\vert))$.

For the correctness, we prove the following invariant for each iteration. Whenever a set is returned in an iteration,
\begin{itemize} \setlength{\itemsep}{-2pt}
\item the set of vertices currently in the forbidden queue are exactly the $i$-forbidden vertices
\item the vertices in the current BFS tree are exactly the vertices in the set $R(Y_{i})$. 
\end{itemize}

\noindent
For the first iteration, this is clearly true. We assume that the invariant holds at the end of iteration $j\geq 1$ (where $i=i'$) and consider the $(j+1)$-th iteration (where $i$ is now set as $i'-1$).

Let $P_j$ be the vertices present in the BFS tree at the end of the $j$-th iteration and $P_{j+1}$ be the vertices present in the BFS tree at the end of the $(j+1)$-th iteration. We claim that the set $P_{j+1}=R(Y_{i'-1})$.
 
Since we never add a vertex to $P_{j+1}$ if it is $(i'-1)$-forbidden, the vertices in $P_{j+1}\setminus P_j$ are precisely those vertices which are reachable from $P_j$ via a path disjoint from $(i'-1)$-forbidden vertices.
Since the invariant holds for the preceding iteration, we know that $P_j=R(Y_{i'})$ and by our observation about $P_{j+1}\setminus P_j$, we have that $P_{j+1}$ is the set of vertices reachable from $R(Y_{i'})$ via paths disjoint from $(i'-1)$-forbidden vertices, which implies that $P_{j+1}=R(Y_{i'-1})$ since $R(Y_{i'-1})$ is precisely the set of vertices reachable from $R(Y_{i'})$ via paths disjoint from $(i'-1)$-forbidden vertices. We now show that the vertices in the forbidden queue are exactly the $(i'-1)$-forbidden vertices.
Since the BFS tree in iteration $j+1$ could not be grown any further, every out-neighbor of every vertex in the tree is in the forbidden queue. Since we have already shown that the vertices in the BFS tree, that is in $P_{j+1}$, are precisely the vertices in $R(Y_{i'-1})$, we have that every $(i'-1)$-forbidden vertex is already in the forbidden queue.
This proves that the invariant holds in this iteration as well and completes the proof of correctness of the algorithm.
\end{proof}


The main lemma of this section is the following. 

\begin{lemma}\label{lem:find imp set skewsym}
Let $D=((V,A),\sigma)$ be a skew-symmetric graph and $k\in \mathbb{N}$. Let $L\subseteq V$ be a regular set of vertices such that there is an $L$-$L'$ path in $D$.  
There is an algorithm which runs in time $\bigoh(k^3(m+n))$ and 
\begin{itemize}
\setlength{\itemsep}{-1pt}
 \item correctly concludes that no $(L,k)$-component exists or
 \item returns an $(L,k)$-component or
 \item returns an irregular minimum $L$-$L'$ separator
\end{itemize}
where $m=\vert A\vert$ and $n=\vert V\vert$.
\end{lemma}

\begin{proof} 
The main idea of the algorithm is to start with the collection coming from Lemma~\ref{lem:separator layer} and then use this to either find an 
$(L,k)$-component or to return an irregular minimum $L$-$L'$ separator. In what follows we describe possible situations that could arise and how they could be handled. Finally, we use all this to describe the algorithm and prove its correctness. 

We can, in $\bigoh(k(m+n))$ time check if the size of the minimum $L$-$L'$ separator is at most $2k$ by running $2k$ iterations of the Ford-Fulkerson algorithm~\cite{FordFulkerson56}. If the size of the minimum separator exceeds $2k$, then we can conclude that no $(L,k)$-component exists. Therefore, we assume in the rest of this proof that the size of the minimum $L$-$L'$ separator is at most $2k$. Let ${\cal X}=\{X_1,\dots,X_q\}$ be a collection with the properties mentioned in Lemma~\ref{lem:separator layer} where ``$L$ acts as $s$ and $L'$ acts as $t$". We make this formal when we describe the algorithm later.

We begin by showing that not all $X_i$'s can be irregular.
\begin{claim}
\label{claim:firstisregular}
 There is an index $i\geq 1$ such that for all $j\leq i$, the set $X_j$ is regular.
\end{claim}
\begin{proof}
 We first show that $X_1$ is regular. Suppose that this is not the case and there are vertices $y,y'\in X_1$. Since no arc in $A[X_1]$ is part of a minimum $L$-$L'$ separator (by property 4 of the collection), Lemma~\ref{lem:lk component} implies that no arc in $A[X_1]$ is the conjugate of an arc in $S=\delta^+(X_1)$. Therefore, there is a path from $L$ to $y$ and a path from $L$ to $y'$ disjoint from $S\cup S'$. However, this implies the presence of a path from $L$ to $L'$ in $D\setminus (S\cup S')$, which is a contradiction. Therefore, $X_1$ is regular.
Since every subset of a regular set is regular, there is an index $i\geq 1$ such that for all $j\leq i$, $X_j$ is regular.
This completes the proof of the claim.
\end{proof}

Given the above claim, we first consider the case when $X_i$ is regular for every $1\leq i\leq q$. This brings us to the following claim. 
\begin{claim}
\label{claim:qiscomponent}
If $X_i$ is regular for every $1\leq i\leq q$ then $X_q$ itself is an $(L,k)$-component. 
\end{claim}
\begin{proof}
 Observe that $X_q$ satisfies the first four properties of an $(L,k)$-component and thus it suffices to prove that $X_q$ is inclusion-wise maximal with respect to these 4 properties. Suppose that $X_q$ is not maximal with respect to these properties and let $Z\supset X_q$ have the required properties and let $Y$ be a minimum $Z$-$Z'$ separator. Clearly, $Y$ is a minimum $L$-$L'$ separator. Since $Z\supset X_q$, there is an arc $y\in Y$ which is not in $\delta^+(X_q)$. Since $X_q$ strictly contains all other $X_i$'s, $y\notin \delta^+(X_i)$ for any $i$, which contradicts property 4 of the collection.
\end{proof}

Claims \ref{claim:firstisregular} and \ref{claim:qiscomponent} act like  base cases of our algorithm  that we describe later.


\medskip

We now suppose that there exists $i\leq q$ such that $X_i$  is irregular. 
Let $a$ be the highest index such that $X_a$ is regular and $X_{a+1}$ is irregular. Let $A=X_a$ and $B=X_{a+1}$. Since $\delta^+(B)$ is a minimum $L$-$L'$ separator, by the crossing-uncrossing lemma (Lemma~\ref{lem:crossing}), we have that $\vert \delta^+(B\setminus B')\vert=\vert \delta^+(B)\vert$ and $\delta^+(B)\cap \delta^+(B')=\emptyset$. That is, there is no arc which enters $V\setminus (B\cup B')$ from $B\cap B'$ and thus there is no arc which 
enters $B\cap B'$ from $V\setminus (B\cup B')$. 
Furthermore, if there is an arc $x\in \delta^+(B\setminus B')\cap \delta^-(B'\setminus B) $, then $x'\in \delta^+(B\setminus B')$, which implies that $x,x'$ are contained in a minimum $L$-$L'$ separator. Therefore, from this point on, we may assume that there is no arc in $\delta^+(B\setminus B')\cap \delta^-(B'\setminus B)$. Before we go further we summarize the sets and the various intersections they have. From now onwards the sets $B$ and $Q=B\setminus B'$ will always have the following intersection properties. 

 \medskip
  \begin{center}
 \begin{boxedminipage}{.55\textwidth}
 \small{
\begin{enumerate}
\setlength{\itemsep}{-2pt}
 \item $Q=B\setminus B'$ 
\item $\vert \delta^+(B\setminus B')\vert=\vert \delta^+(Q)\vert=\vert \delta^+(B)\vert$  
\item $\delta^+(B)\cap \delta^+(B')=\emptyset$
\item $\delta^-(B)\cap \delta^-(B')=\emptyset$
\item $\delta^+(B\setminus B')\cap \delta^-(B'\setminus B)$ =    $\delta^+(Q)\cap \delta^-(Q')=\emptyset$
\end{enumerate}
%
 }
\end{boxedminipage}
 \end{center}


\begin{figure}[t] \begin{center}
\includegraphics[height=200 pt, width=230 pt]{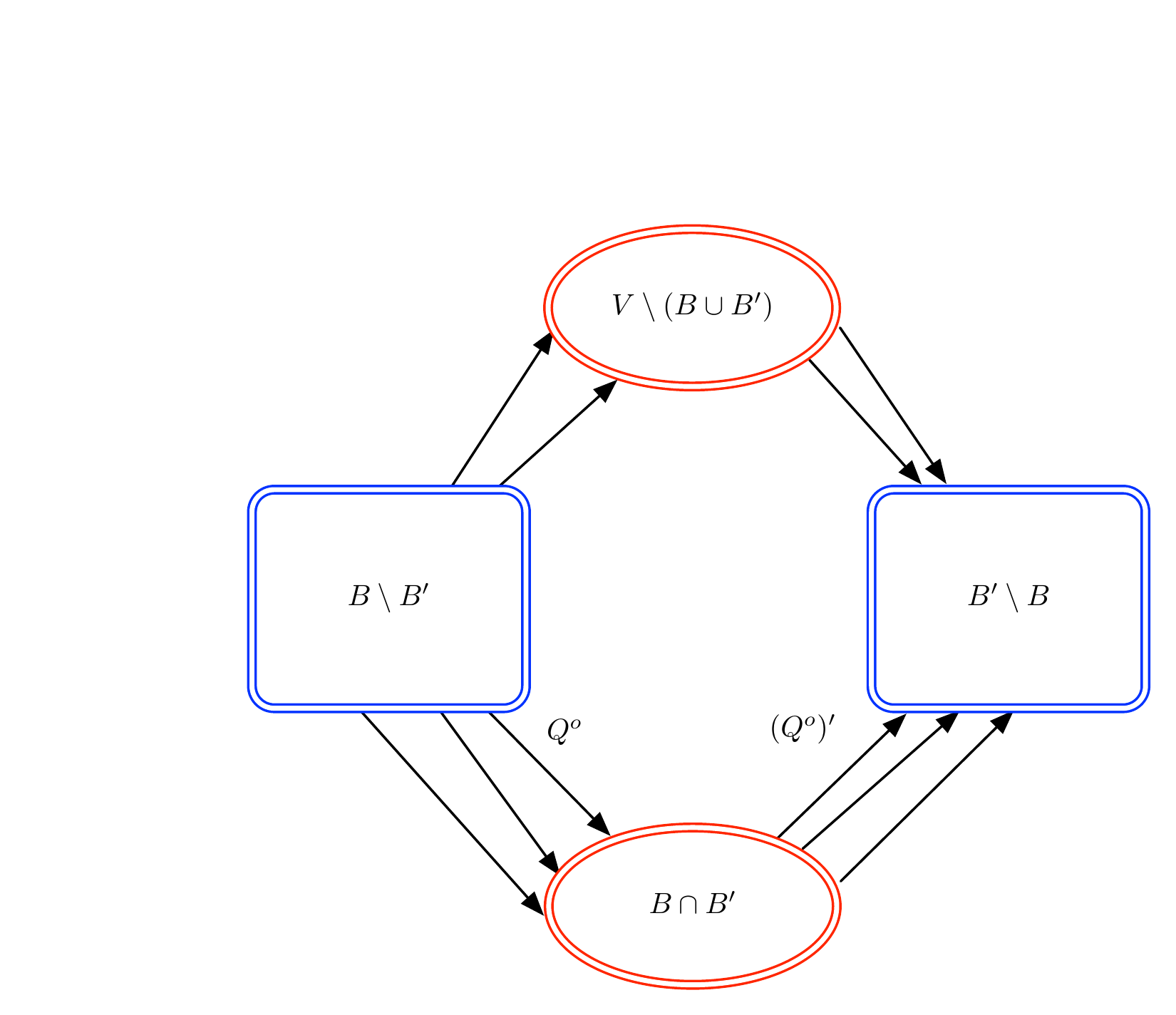}

 \caption{An illustration of the sets defined in Lemma~\ref{lem:find imp set skewsym}. Observe that there are no arcs between $B\setminus B'$ and $B'\setminus B$ and between $B\cap B'$ and $V\setminus (B\cup B')$}
\label{fig:crossing2}
\end{center}
\end{figure}

\noindent 
The proof of the next observation follows from the fact that there is neither an arc entering $B\cap B'$ from $V\setminus (B\cup B')$ nor an 
are leaving $B\cap B'$ to $V\setminus (B\cup B')$ (see Figure~\ref{fig:crossing2}). 
\begin{observation}
\label{obs:disjointpath}
 Any path from $Q$ to $Q'$ with the internal vertices disjoint from $Q\cup Q'$ is contained entirely in either $D\setminus (B\cap B')$ or $D[B\cup B']$.
\end{observation}

The next claim describes certain properties of paths exiting $Q$ and entering $B\cap B'.$ This will be used later in some of the arguments.

\begin{claim}
\label{claim:pathinbcapb}
 Let $B$ and $Q$ be defined as above. Then for every $q\in N^+(Q)\cap (B\cap B')$ there is a path from $q$ to $N^-(Q')\cap 
 (B\cap B')$ which completely lies in the graph $D[B\cap B']$.  
\end{claim}
\begin{proof}
Suppose this is not the case. That is, there exists an edge $x=(p,q)\in \delta^+(Q)$ such that $q\in (B\cap B')$ and there is no path from 
$q$ to $N^-(Q')\cap (B\cap B')$ which completely lies in the graph $D[B\cap B']$.  
Consider the graph $D_1=D\setminus (\delta^+(Q)\setminus \{x\})$. 
Since $\delta^+(Q)$ is a minimum $L$-$L'$ separator, there is  path from $L$ to $L'$ in $D_1$. Since this path contains the arc $x$, there is a 
subpath, say $W$, from $q$ to $N^-(Q')$ that does not intersect $Q'$. Furthermore, in $D_1$ the only arc that emanates from $Q$ is $x$ 
and thus we have that $W$ is disjoint from $Q$ as well. However, by Observation~\ref{obs:disjointpath} every path from $Q$ to $Q'$ with the internal vertices disjoint from $Q$ and $Q'$ is contained in $D[B\cup B']$ or $D[V\setminus (B\cap B')]$. Since $q\in B\cap B'$, we conclude that there is a path from $q$ to $N^-(Q')\cap (B\cap B')$ in $D[B\cap B']$. This is a contradiction to our assumption and thus 
concludes the proof. 
\end{proof}

\noindent
We are now ready to describe the cases that occur when we  {\em have} an irregular set. This case is divided into following three exhaustive subcases. 
\begin{description}
\setlength{\itemsep}{-3pt}
 \item[Case I:] $A\cap B'=\emptyset$. 
\item[Case II:] $A\cap B'\neq \emptyset$ and $A\setminus B'\neq \emptyset$. 
\item[Case III:] $A\subseteq B\cap B'$.
\end{description}

\noindent 
We now consider each case one by one and show how we will handle it algorithmically (later).

\begin{description}
\setlength{\itemsep}{-3pt}
 \item[Case I:] $A\cap B'=\emptyset$. We start by defining the set $Q^{o}=\delta^+(Q)\cap \delta^-(B\cap B')$.

 \end{description}

\noindent 
The next claim proves an interesting structural property that is crucial to the correctness of the algorithm. 
\begin{claim}
Either every $(L,k)$-component $Z\supseteq Q$ is such that $Q^{o}\subseteq \delta^+(Z)$ or there is an arc $y\in Q^o$ such that there is a minimum $L$-$L'$ separator containing $y$ and $y'$.
\end{claim}
\begin{proof}
Let $Z\supseteq Q$ be an $(L,k)$-component and $x\in Q^o$ an arc such that $x=(p,q)\notin \delta^+(Z)$. Note that this implies that $x\in A[Z]$. 
By Claim~\ref{claim:pathinbcapb}, we have that there is a path from $q$ to $N^-(Q')\cap (B\cap B')$ which lies in the graph $D[B\cap B']$.  
Observe that $N^-(Q')\cap (B\cap B')=\tail((Q^o)')$. 
Let $r\in B\cap B'$ be such that $q$ has a path in $D[B\cap B']$ from $q$ to $r$, where $(r,s)\in (Q^o)'$. We claim that there is a minimum $L$-$L'$ separator containing $(r,s)$ and $(s',r')$. Consider a minimum $Z$-$Z'$ separator $S$. Since every arc in $A[B\cap B']$ has both endpoints inside $B$ and both endpoints outside $A$,  none of these arcs appear in the set $\bigcup_{i=1}^q\delta^+(X_i)$, we conclude that $S$ is disjoint from $A[B\cap B']$ (by property 4 of the collection). Since $S$ is disjoint from $A[B\cap B']$, we have that $r\in Z$. Furthermore, since $s\in Q$, we have that $s\in Z'$. Consequently, we have that $r'\in Z'$ and $s'\in Z$. Therefore, both the arcs $(r,s)$ and $(s',r')$ are in both the sets $\delta^+(Z)$ and $\delta^-(Z')$, which implies that they are also present in $S$. Since $S$ is also a minimum $L$-$L'$ separator, this completes the proof of the claim.   
\end{proof}


\vspace{5 pt}
\noindent
The following claims allow our algorithm to use the above observations recursively in the graph $D\setminus (B\cap B')$ in the absence of $y\in Q^o$ such that there is no minimum $L$-$L'$ separator containing $y$ and its conjugate $y'$. 

%
\begin{claim}
Assume that every $(L,k)$-component $Z\supseteq Q$ is such that $Q^{o}\subseteq \delta^+(Z)$. 
(Recall $Q^{o}=\delta^+(Q)\cap \delta^-(B\cap B')$.)
\begin{enumerate}
\setlength{\itemsep}{-2pt}
\item If $S$ is a minimum $Q$-$Q'$ separator in $D_1=D[V\setminus (B\cap B')]$, then $S\cup Q^o$ is a minimum $Q$-$Q'$ separator in $D$.
\item If $Z$ is an $(L,k)$-component such that $Q^{o}\subseteq \delta^+(Z)$, 
then $Z$ is also an $(L,k-\vert Q^o\vert)$-component in $D_1=D[V\setminus (B\cap B')]$ and furthermore, any $(L,k-\vert Q^o\vert)$-component in $D_1$ is an $(L,k)$-component in $D$.
\end{enumerate}
 \end{claim}
%

\begin{proof}
Recall that  $\delta^+(B)\cap \delta^+(B')=\emptyset$ and $\delta^-(B)\cap \delta^-(B')=\emptyset$. Hence every $Q$-$Q'$ path in $D$ is contained entirely in the graph $D_1$ or $D_2=D[B\cup B']$. The last assertion implies that the size of the minimum  $Q$-$Q'$ separator in $D$ is the sum of the sizes of the minimum $Q$-$Q'$ separator in the graphs $D_1$ and $D_2$.  Since $\delta^+(Q)\setminus Q^o$ is a minimum $Q$-$Q'$ separator in $D_1$, $S$ is no larger than $\delta^+(Q)\setminus Q^o$. Therefore, $S\cup Q_o$ is no larger than $\delta^+(Q)$. Since $S\cup Q^o$ is clearly a $Q$-$Q'$ separator in $D$, this completes the proof of this statement.

Now we show the second part of the claim. We first show that $Z$ satisfies the first 4 properties of an $(L,k-\vert Q^o\vert)$-component in $D_1$.  Recall that the size of a minimum $Q$-$Q'$ separator in the graph $D_1$ is $\vert \delta^+(Q)\vert -\vert Q^o\vert$, which implies that $Z$ indeed satisfies the first 4 properties of an $(L,k-\vert Q^o\vert)$-component in $D_1$.

 Therefore, if $Z$ were not an $(L,k-\vert Q^o\vert)$-component in $D_1$, then there is a set $W\supset Z$ which satisfies these 4 properties and let $S$ be a minimum $W$-$W'$ separator in $D_1$. We claim that $W$ also satisfies the first 4 properties of an $(L,k)$-component in $D$, which contradicts our assumption of $Z$ as an $(L,k)$-component. From the first part of the claim, we have that $S\cup Q^o$ is a minimum $Q$-$Q'$ separator in $D$, which implies that $W$ indeed satisfies the first 4 properties of an $(L,k)$-component in $D$.

In the converse direction, let $Z$ be an $(L,k-\vert Q^o\vert)$-component in $D_1$. Since $S\cup Q^o$ is a minimum $Q$-$Q'$ separator in $D$, 
we have that $Z$ satisfies the first 4 properties of an $(L,k)$-component in $D$. For contradiction assume that $Z$ is not an $(L,k)$-component in $D$. Then there exists $W\supset Z$ that is a $(L,k)$ component in $D$. Since every $(L,k)$-component $W$ is such that 
$Q^{o}\subseteq \delta^+(W)$, we have that $W \cap \head(Q^{o}) =\emptyset$. But then it implies that  $W$ also satisfies the first 4 properties of 
an $(L,k-\vert Q^o\vert)$-component in $D_1$. However this is a contradiction to  our assumption that $Z$ is an $(L,k-\vert Q^o\vert)$-component in $D_1$. This completes the proof of the claim.
\end{proof} 

\noindent
This completes the study of the case when $A\cap B'=\emptyset$.
\vspace{10 pt}

\begin{description}
\item [Case II :]  $A\cap B'\neq \emptyset$ and $A\setminus B'\neq \emptyset$.
\end{description}
\noindent
Here, for a subcase we construct a new collection ${\cal X}'$ where this case is avoided.
Let $Q=B\setminus B'$ and $P=A\cup (B\setminus B')$. We have already observed that $\vert \delta^+(Q)\vert=\vert \delta^+(B)\vert$ and that the set $\delta^+(B)\cap \delta^+(B')$ is empty. We now show that $\vert \delta^+(Q)\vert=\vert \delta^+(P)\vert$.
Let $P_1^o=\delta^+(P)\setminus \delta^+(Q)$ and $Q_1^o=\delta^+(Q)\setminus \delta^+(P)$. We claim that $\vert P_1^o\vert=\vert Q_1^o\vert$. But this is true since $\vert P_1^o\vert <\vert Q_1^o\vert $ implies that $\vert \delta^+(P)\vert <\vert \delta^+(Q)\vert $, which is a contradiction since $\delta^+(P)$ is an $L$-$L'$ separator smaller than the minimum one and $\vert Q_1^o\vert <\vert P_1^o\vert $ implies that $\vert \delta^+(A\setminus B')\vert <\vert \delta^+(A)\vert $, which is a contradiction since $\delta^+(A\setminus B)$ is an $L$-$L'$ separator smaller than the minimum one. Therefore, we conclude that $\vert P_1^o\vert=\vert Q_1^o\vert$ and hence $\vert \delta^+(Q)\vert=\vert \delta^+(P)\vert$. 

By combining this with an application of the crossing-uncrossing lemma (Lemma~\ref{lem:crossing}) on the set $K=B\setminus A'$, we have that 
$\vert \delta^+(P)\vert = \vert \delta^+(K\setminus K')\vert=\vert \delta^+(K)\vert$. Furthermore, $A\subset K\subset B$, and $A\cap K'=\emptyset$.  

If $K$ is irregular, then consider the collection ${\cal X}'$ obtained by inserting the set $K$ between $X_a$ and $X_{a+1}$ in the collection $\cal X$. Clearly, ${\cal X}'$ is also a collection which satisfies the properties 1, 2, and 4 of Lemma~\ref{lem:separator layer}. It also satisfies property 3 since $\delta^+(K)$ is a minimum $L$-$L'$ separator. However, if we consider the collection ${\cal X}'$, $A$ would still be the regular set with the highest index, while $K$ would be the irregular set with the least index. Since $A$ is disjoint from $K'$, we fall back into the previous case when we consider this collection.

If $K$ is regular, then $(B\cap B')\setminus (A\cup A')=\emptyset$.  Observe that $N^+(Q)\cap (B\cap B') \subseteq A\cap B'$ and 
$N^-(Q')\cap  (B\cap B') \subseteq A'\cap B$. Since $A\cap B' \neq \emptyset$ and every vertex in $A$ is reachable from $L$, we have that $N^+(Q)\cap (B\cap B') \neq \emptyset$ and thus $N^-(Q')\cap  (B\cap B') \neq \emptyset$. This together with 
Claim~\ref{claim:pathinbcapb} 
implies that there is a path, say $W$, from $A\cap B'$ to $A'\cap B$ which lies in the graph $D[B\cap B']$. Let $p$ be the last vertex on $W$ from 
$A\cap B'$. Since $A\cap B'$ and $A'\cap B$ partition $B\cap B'$, either there is an arc $(p,q)$ such that $p\in A\cap B'$ and $q\in A'\cap B$ or  
$p\in A\cap B'$ and $q\in B'$. In either case this implies that there is an arc $(q',p')$ where $q'\in A\cap B'$  in the former case and $q\in B'$ in the latter case and $p'\in A'\cap B$. Since both $(p,q)$ and $(q',p')$ are in $\delta^+(K)\cap \delta^-(K')$, and $\delta^+(K)$ is a minimum $L$-$L'$ separator, we have that $\delta^+(K)$ is a minimum $L$-$L'$ separator containing arcs $y$ and $y'$ where $y=(p,q)$.

\medskip
\begin{description}
\item[Case III:] $A\subseteq B\cap B'$. This case is non-existent since $L\subseteq A$ and $L\cap B'=\emptyset$. 

\end{description}


\medskip

\myparagraph{The Algorithm}
We begin by applying the Ford Fulkerson algorithm to compute a minimum $L$-$L'$ separator in the graph. If we require more than $2k$ iterations of the Ford Fulkerson algorithm, then we return that there is no $(L,k)$-component. We then apply the algorithm of Lemma~\ref{lem:separator layer} to compute the collection ${\cal X}=\{X_1,\dots,X_q\}$ which can be computed in time $\bigoh(k(m+n))$. In order to apply Lemma~\ref{lem:separator layer}, we need vertices $s$ and $t$. Therefore, we add vertices $s$ and $t$ to the graph, add $2k+1$ arcs from $s$ to each vertex in $L$ and $2k+1$ arcs from each vertex in $L$ to $t$. It is clear from this construction that no $s$-$t$ separator of size at most $2k$ will contain any of these newly added arcs and therefore, the $s$-$t$ separators of size at most $2k$ are in one to one correspondence with the $L$-$L'$ separators of size at most $2k$.
This allows us to use Lemma~\ref{lem:separator layer} in the form it is stated in.

We then simply need to examine each $X_{i+1}\setminus X_i$ once to compute the index $a$ such that $X_a$ is the highest regular set and $X_{a+1}$ is the least irregular set. After computing the index $a$, in $\bigoh(m)$ time, we can compute the case we are currently in by computing the intersection of the sets $A$ and $B$. If we are in case (b), in time $\bigoh(k^2m)$, we can compute $K$ and either find an irregular minimum $L$-$L'$ separator by testing if there is a minimum $L$-$L'$ separator containing $y,y'$ for some $y\in \delta^+(K)$ or move to case (a) where we already have computed the required sets-- the regular set with the highest index $A$ and the irregular set with the least index $K$. Finally, if we are in case (a), in $\bigoh(m)$ time, we compute the set $Q$ and iteratively compute a $(Q,k')$-component in $D\setminus (B\cap B')$ (which we have already shown is an $(L,k)$-component)  where $k'<k$ or an irregular minimum $Q$-$Q'$ separator. In the latter case, we add $Q^o$ to this minimum separator to get an irregular minimum $Q$-$Q'$ separator in $D$.
The correctness of our algorithm follows from the structural claims preceding the description. Furthermore, each time we iterate, we attempt to compute an $(L,k')$-component where $k'<k$. Therefore,  we can have at most $2k$ such iterations and hence, the running time of our algorithm is bounded by $\bigoh(k^3(m+n))$. This completes the proof of the lemma.
\end{proof}


\section{Algorithm for {\skewsymmc}}\label{sec:algo}
In this section we design our linear time parameterized algorithm for {\skewsymmc}. We first give a lemma which allows us to find a solution that is 
disjoint from some part of the solution.  More formally we show the following. 

\begin{lemma}\label{lem:main branching skewsym}
Let $(D=(V,A),\sigma,{\cal T},k)$ be a {\Yes} instance of {\skewsymmc} and let $L$ be a regular set of vertices such that there is an $L$-$L'$ path in $D$. If there is a solution for the given instance which is an $L$-$L'$ self-conjugate separator in $D$, then the following hold. 

\begin{enumerate}

\item An $(L,k)$-component exists.

\item Let $Z\subseteq V$ be a regular set of vertices containing $L$ such that the size of the minimum $Z$-$Z'$ separator is the same as the size of the minimum $L$-$L'$ separator. Then, there is a solution for the given instance disjoint from $A[Z]$.

\end{enumerate}

\end{lemma}

\begin{proof}

For the first statement, observe that since there is a solution for the given instance which is an $L$-$L'$ self-conjugate separator, the size of the minimum $L$-$L'$ separator is at most $2k$. Therefore the set $L$ itself satisfies the first 4 properties of an $(L,k)$-component and therefore an $(L,k)$-component exists. This completes the proof for the first statement.

 \begin{figure}[t] \begin{center}
\includegraphics[height=140 pt, width=370 pt]{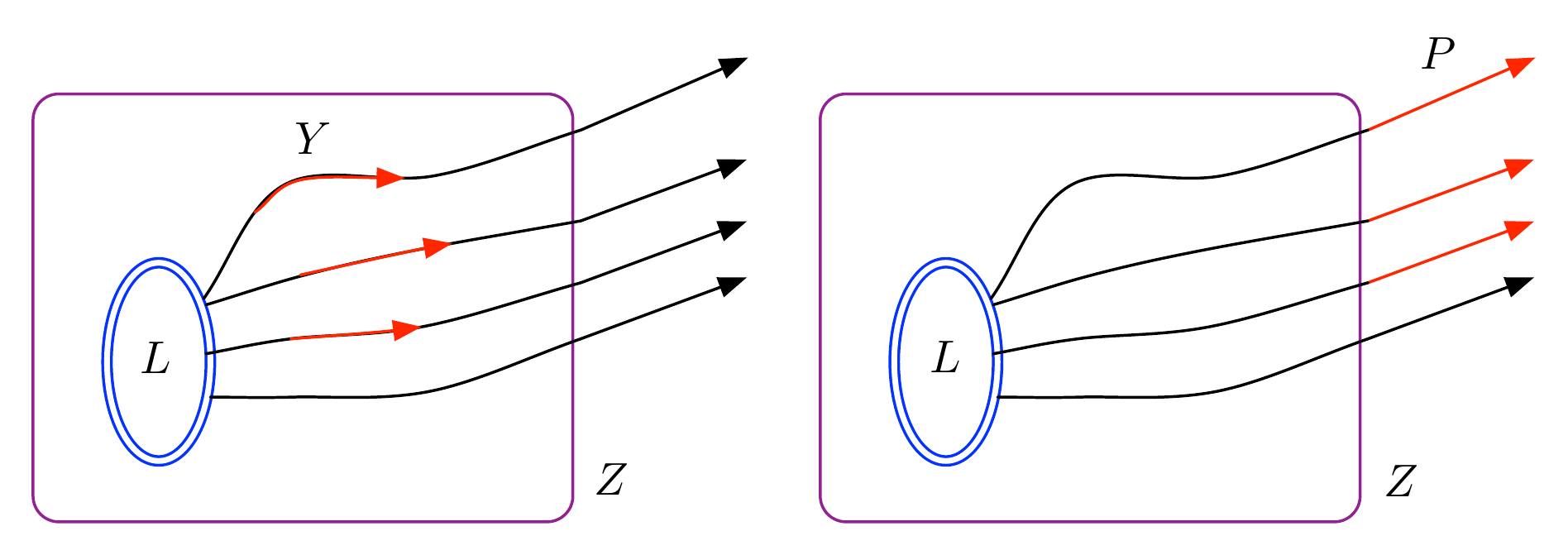}

 \caption{An illustration of the sets $Z$, $Y$ and $P$ in the proof of Lemma~\ref{lem:main branching skewsym}.}
\label{fig:disjointsolution}
\end{center}
\end{figure}

\smallskip
Now we prove the second part of the lemma. 
Let $X$ be the solution defined above, that is, $X$ is an $L$-$L'$ self-conjugate separator. If $X$ is disjoint from $A[Z]$, then we are done. Therefore, suppose that $X$ intersects $A[Z]$ and let $Y$ be the set of arcs in $A[Z]$ such that $Y \cup Y'\subseteq X$.  Let 
$P\subseteq \delta^+(Z)$ be such that $\tail(P)$ is not reachable from $L$ in $D[Z]\setminus Y$. That is,  $Y$ is an $L$-$\tail(P)$ 
separator in the digraph $D[Z]$.  It might be possible that  $\tail(P)=\emptyset$. 
 We now claim that the set $\widehat X=(X\setminus (Y\cup Y'))\cup (P\cup P'))$ (see Figure~\ref{fig:disjointsolution}) is also a solution for the given instance.
 \begin{claim}
  $\widehat X=(X\setminus (Y\cup Y'))\cup (P\cup P'))$ is a solution for the given instance and $|\widehat{X}|\leq |X|$.
 \end{claim}
\begin{proof}
We first show that $|\widehat{X}|\leq |X|$. 
Since $\delta^+(Z)$ is a minimum $L$-$L'$ separator, there are $\vert \delta^+(Z)\vert$ arc disjoint paths from $L$ to $\tail(\delta^+(Z))$ in $D[Z]$. Therefore, $\vert Y\vert \geq \vert P\vert$. Clearly, $\widehat X$ is no larger than $X$. Therefore, it remains to show that $\widehat X$ is a skew-symmetric multicut for $D$. If this were not the case, then there is a closed walk in the graph $D\setminus \widehat X$ intersecting an arc $y\in Y$ and containing vertices $t$ and $t'$. Here $t\in J$ where $J\in \cal T$ and for every vertex $v_i \in J$,  $v_i$ and $\sigma(v_i)$ lie in the same strongly connected components of $D\setminus \widehat{X}$.  That is, $J$ is a violated constraint.

Since $\tail(y)$ is reachable from $L$ in the graph $D\setminus \widehat X$, both $t$ and $t'$ are reachable from $L$ in the graph $D\setminus \widehat X$, which implies the presence of a path, say $W$, from $L$ to $L'$ in $D\setminus \widehat X$. Now using this path we construct another path $W'$ from  $L$ to $L'$ in $D\setminus X$. This will contradict our assumption that $X$ is an $L$-$L'$ self-conjugate separator in $D$.

Observe that  $W$ must also intersect an arc in $\delta^+(Z)$ since $L\subseteq Z$ and $L'$ is disjoint from $Z$. However, since $P\cup P'\subseteq \widehat X$, this arc, say $(p,z_1)$,  is in $\delta^+(Z)\setminus P$. Furthermore, this path also contains a subpath from a vertex $z_1\in \head(\delta^+(Z)\setminus P)$ to a vertex $z_2\in \tail(\delta^-(Z')\setminus P')$ whose arcs are disjoint from $A[Z\cup Z']\cup \delta^+(Z)\cup \delta^-(Z')$.  We call this path $W_{z_1z_2}$.  Let $(z_2,q)$ be an arc in $\delta^-(Z')\setminus P'$ such that the arc $(z_2,q) $ is on $W$. Observe that $p$, $q'$ are in $Z$ and there are paths from $L$ to both $p$ and $q'$ that avoids arcs of $Y$. The last assertion follows from the fact that 
the vertices of $ \delta^+(Z)\setminus P$ are  reachable from $L$ in $D[Z]\setminus Y$.  Let these paths avoiding arcs of $Y$ be called $W_{Lp}$ and $W_{Lq'}$. Then observe that $W_{Lp}W_{z_1z_2}(W_{Lq'})'$ forms a path in $D\setminus X$ -- a contradiction. Here, $(W_{Lq'})'$  is the path that is conjugate to $W_{Lq'}$. This completes the proof of the claim. 
%
%
%
%
%
%
%
%
%
\end{proof}
The above claim  completes the proof of  the lemma. 
\end{proof}

\noindent
From this point on, we assume that an instance of {\skewsymmc} is of the form $(D=(V,A),\sigma,T,k,L)$ where $L$ is a regular set of vertices and the question is to check {\em if there is a solution for the given instance which is an $L$-$L'$ self-conjugate separator.} To solve the problem on the given input instance, we simply solve it on the instance $(D=(V,A),\sigma,{\cal T},k,\emptyset)$.

\begin{lemma}\label{lem:rule correctness} Let $(D=(V,A),\sigma,{\cal T},k,L)$ be an instance of {\skewsymmc} and let $S$ be an irregular minimum $L$-$L'$ separator. Then, there are arcs $y,y'$ such that  $y,y'\in \delta^+(R(L,S\cup S')))$ and there is a solution for the given instance containing $y,y'$.  
\end{lemma}

\begin{proof} Let $Z=R(L,S\cup S')$. By Lemma~\ref{lem:lk component}, $\delta^+(Z)$ is a minimum $L$-$L'$ separator and hence $\vert \delta^+(Z)\vert=\vert S\vert$. Since $S$ is irregular, $S\cup S'$ contains arcs from at most $\vert S\vert-1$ conjugate pairs. Thus  
there are arcs $y,y'\in \delta^+(Z)$.

Let $X$ be a solution for the given instance. Since there is no path from $L$-$L'$ in $D\setminus X$, it must be the case that $D\setminus X$ cannot contain paths from $L$ to both $\tail(y)$ and $\tail(y')$. Therefore, it must be the case that $X$ intersects $A[Z]$ and intersects all paths from $L$ to $\tail(y)$ or $\tail(y')$. However, by Lemma~\ref{lem:main branching skewsym} we know that there exists a solution that does not intersect $A[Z]$. This implies that there is also a solution containing the arcs $y,y'$.
\end{proof}

\noindent
The above lemma gives us the following reduction rule.

\begin{redrule}\label{red:regularity}
 Let $(D=(V,A),\sigma,T,k,L)$ be an instance of {\skewsymmc} and let $S$ be an irregular minimum $L$-$L'$ separator. Let $y,y'$ be the arcs given by Lemma~\ref{lem:rule correctness}. Then return the instance $(D=(V,A\setminus \{y,y'\}),\sigma,{\cal T},k-1)$.  
\end{redrule}

\medskip

\noindent
Therefore, by combining this reduction rule with the algorithm of Lemma~\ref{lem:find imp set skewsym}, we can, in linear time either reduce the parameter or compute an $(L,k)$-component with a regular neighborhood.
We are now ready to prove Theorem~\ref{thm:skewsymmc} by giving an algorithm for {\skewsymmc}.

\medskip

\noindent
{\bf Description of Algorithm.} The input to our algorithm for {\skewsymmc} is an instance $(D=(V,A),\sigma,{\cal T}=\{J_1,\dots,J_r\},k,L)$ where $J_i=\{v_{i_1},\dots,v_{i_d}\}$ and 
the algorithm either returns a skew-symmetric multicut of size at most $2k$ which is an $L$-$L'$ self-conjugate separator in $D$ or concludes correctly that no such set exists. In order to solve the problem on the given instance of {\skewsymmc}, the algorithm is invoked on the input $(D=(V,A),\sigma,{\cal T},k,\emptyset)$.

\medskip
\noindent
{\bf 1.} If $L=\emptyset$ or if there is no path from $L$ to $L'$ in $D$, then the algorithm
checks if there is a set $J_i\in {\cal T}$ such that for all $1\leq s\leq d$, $v_{i_s}$ and $v_{i_s}'$ lie in the same strongly connected component in $D$, that is, a violated set. If there is no such set, then the algorithm returns the emptyset. Otherwise, the algorithm picks such a set $J_i$, and branches in $2d$ ways. In the first $d$ branches, it recurses on the instances $\{(D=(V,A),\sigma,{\cal T},k,\{v_{i_j}\})\}_{1\leq j\leq d}$ and in the next $d$ branches, it recurses on the instances $\{(D=(V,A),\sigma,{\cal T},k,\{v_{i_j}'\})\}_{1\leq j\leq d}$.

\medskip
\noindent
{\bf 2.} Suppose $L\neq \emptyset$ and there is an $L$-$L'$ path in $D$. 
Then, the algorithm of Lemma~\ref{lem:find imp set skewsym} is first used on the instance $(D=(V,A),\sigma,{\cal T},k,L)$ to either compute an $(L,k)$-component (if it exists) or an irregular minimum $L$-$L'$ separator. If an $(L,k)$-component does not exist, then we return {\No}. If an irregular minimum $L$-$L'$ separator is returned, then we apply Reduction Rule~\ref{red:regularity}. Suppose that an $(L,k)$-component $Z$ is returned. We check if Reduction Rule~\ref{red:regularity} applies on any arc in $\delta^+(Z)$ and if it does, apply the rule. Therefore, at this point, we may assume that an $(L,k)$-component $Z$ is returned and that the rule is not applicable on any arc in $\delta^+(Z)$.  Observe that $\delta^+(Z)\neq \emptyset $ since there is an $L$-$L'$ path in $D$. The algorithm then picks an arc $a\in \delta^+(Z)$ and branches in 2 ways as follows. In the first branch, the algorithm deletes $\{a,a'\}$ and recurses on the resulting instances, that is, the algorithm recurses on the instance $(D=(V,A\setminus \{a,a'\}),\sigma,{\cal T},k-1,L)$. In the next branch, the algorithm recurses on the instance assuming that $a$ is in $A[R(L,X)]$ where $X$ is the hypothetical solution, that is, the algorithm recurses on the instance $(D=(V,A),\sigma,{\cal T},k,L\cup \{\head(a)\})$. 

\medskip

\noindent
{\bf Correctness.} The correctness of the algorithm is proved by induction on a measure defined on the instance  $I$. Let this measure be denoted 
by  $\mu(I)=\mu((D,k),L)=2k-\lambda(L,L')$ where $\lambda(L, L')$ is the size of the smallest $L$-$ L'$  separator. In the base case, if $\lambda(L, L')>2k$, then the algorithm of Lemma~\ref{lem:find imp set skewsym} returns {\No} on input $(D=(V,A),\sigma,k,L)$, and hence this algorithm returns {\No} as well, which is correct since the solution we require contains an $L$-$L'$ separator of size at most $2k$. Similarly, the case when $k<0$ is also clearly correct. We now assume as induction hypothesis that the algorithm is correct on all instances $I$ such that $\mu(I)\leq \mu-1$ and consider an instance $I=(D=(V,A),\sigma,{\cal T},k,L)$ such that $\mu(I)=\mu$ and $k\geq 0$.  

We first show that an application of Reduction Rule~\ref{red:regularity} does not decrease this measure. 
Since deleting an arc and its conjugate from an irregular minimum $L$-$L'$ separator reduces the size of the minimum size $L$-$L'$ separator by 2 and the budget $k$ by 1, the measure $2k-\lambda(L,L')$ remains unchanged.

We now consider the branching rules. If $L=\emptyset$ or there is no $L$-$L'$ path in $D$, then $\lambda (L,L')=0$. Consider an instance $I'=(D=(V,A),\sigma,{\cal T},k,\{v\})$ resulting from a branch here. Although the parameter has not decreased here, since there is a path from $v$ to $v'$, $\lambda(v,v')>0$, which implies that $\mu(I')<\mu(I)$. Therefore, by combining the exhaustiveness of the branching along with the induction hypothesis, we obtain the correctness of the algorithm on the instance $I$ as well.

We now suppose that $L\neq \emptyset$ and there is an $L$-$L'$ path in $D$.
The branching is exhaustive due to Lemma~\ref{lem:main branching skewsym}. 
We now show that for each of the resulting instances $I'$ from any of the branches, $\mu(I')\leq \mu-1$ in which case we can apply the induction hypothesis on these instances, thus proving the correctness of the algorithm.\\

\noindent
{\bf (a)} We begin with the branch where we recurse on the instance $I'=(D=(V,A),\sigma,{\cal T},k,Z\cup \{\head(a)\})$. Since $Z$ is an $(L,k)$-component, by definition, the size of the smallest separator from $Z\cup \{\head(a)\}$ to $Z'\cup \{(\head(a))'\}$ is strictly larger than $\lambda(L,L')$, which implies that $\mu(I')<\mu(I)$. 

\medskip
\noindent
{\bf (b)} We now consider the instance resulting from the other branch, that is the instance $I'=(D=(V,A\setminus \{a,a'\}),\sigma,{\cal T},k-1,L)$ where $a\in \delta^+(Z)$. In this case, the parameter has decreased by 1. Since Reduction Rule~\ref{red:regularity} is not applicable on $\delta^+(Z)$, removing $\{a,a'\}$ from the graph reduces $\lambda(L,L')$ by at most $1$ and therefore, we have that $\mu(I')<\mu$. This completes the proof of correctness of the algorithm.

\begin{figure}[t] \begin{center}
\includegraphics[height=200 pt, width=230 pt]{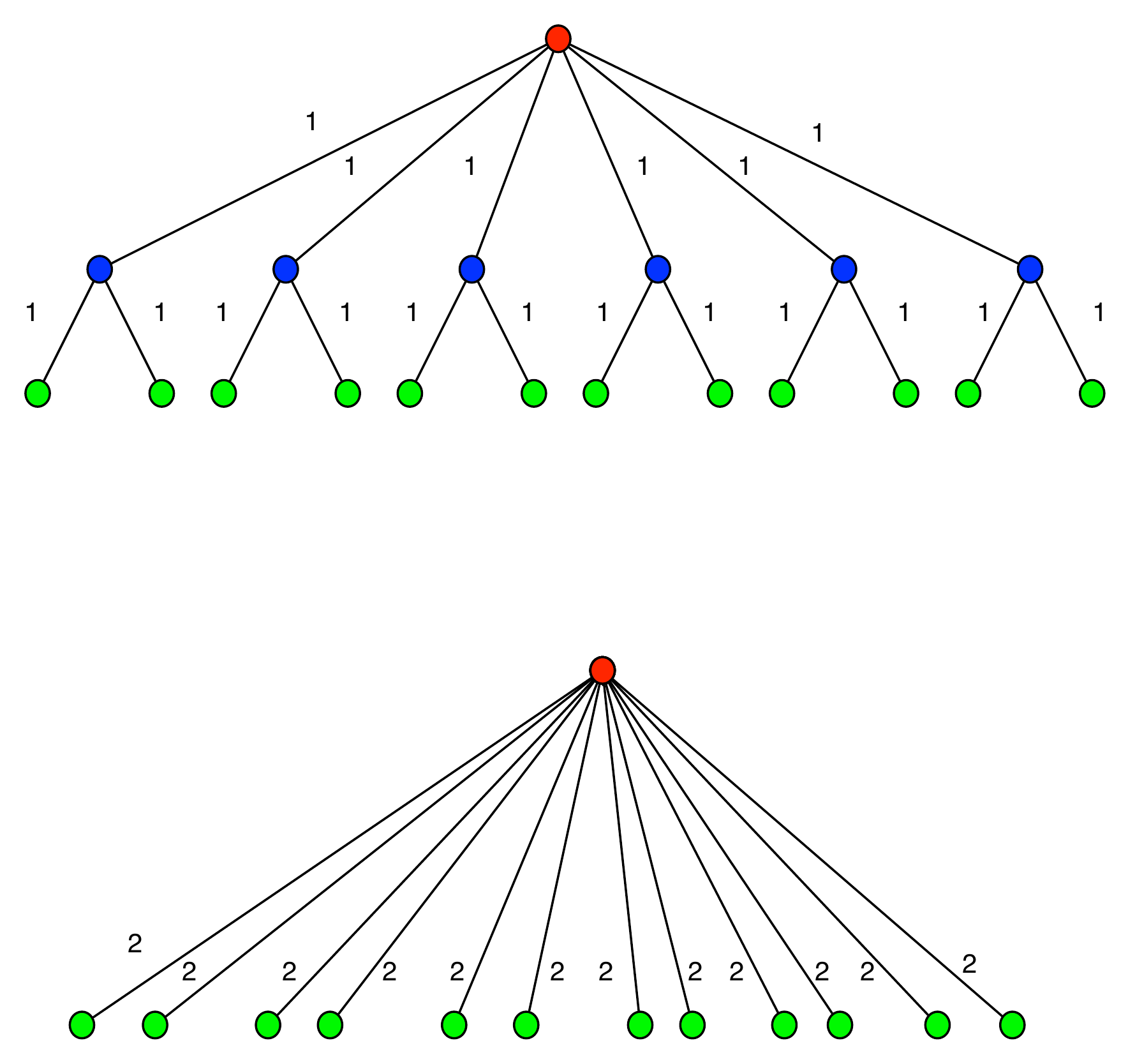}

 \caption{An illustration of the tighter analysis of the search tree in the algorithm.}
\label{fig:searchtree}
\end{center}
\end{figure}

\medskip
\noindent
{\bf Running time.} 
We prove that on an instance $I=(D=(V,A),\sigma,{\cal T},k,L)$, the algorithm computes a search tree with at most  $(2\sqrt{d})^{\mu(I)}$ leaves. 
We have already proved that in each branch, the measure $\mu(I)$ decreases by at least 1 and since we only have $2d$-way branchings, the number of nodes of the search tree is clearly bounded by $(2d)^{\mu(I)}$. However, we analyze more closely a branch which occurs in a $2d$-way branching where $\mu(I)$ decreases by exactly 1. Suppose that $J$ was the violating set computed in this step and suppose $v\in J$ be such that $\lambda(v,v')=1$. Consider the branch where we recurse on the instance $(D=(V,A),\sigma,{\cal T},k,\{v\})$. In this recursion, we first observe that the reduction rule will not be applied. This is because the reduction rule requires a minimum $\{v\}$-$\{v'\}$ separator of size at least 2 while $\lambda(v,v')=1$. Therefore, the algorithm of Lemma~\ref{lem:find imp set skewsym} will find a $(\{v\},k)$-component $Z$.  
\medskip

\begin{claim} $\vert\delta^+(Z)\vert=1$. \end{claim}

\begin{proof}
  Suppose not and let $(a,b),(p,q)\in \delta^+(Z)$. Furthermore, let $S$ be a minimum $Z$-$Z'$ separator. Then, since $\vert S\vert=1$ by our assumption, either $(a,b)\notin S$ or $(p,q)\notin S$. Suppose that $(a,b)\notin S$. Then, $S$ is also a minimum $Z\cup \{b\}$-$Z'\cup \{b'\}$ separator, which contradicts the maximality of $Z$ as a $(\{v\},k)$-component.
  \end{proof}
  
  \noindent
 Consider the branching performed on the single arc in $\delta^+(Z)$. We have already shown that the measure decreases by 1 in each of these branchings. Therefore, we combine these 2 branches with the branch where we decided to recurse on the instance $(D=(V,A),\sigma,{\cal T},k,\{v\})$. This leads to 2 branches where the measure decreases by 2 in each. For each branch in the $2d$-way branching where the measure decreases by 1, we can do the same to obtain at most $4d$ branches in each of which the measure decreases by 2 (see Figure~\ref{fig:searchtree}).
Therefore, our worst case branching is a $4d$-way branching, where the measure drops by 2 in each branch. Thus, we obtain a bound of $(2\sqrt{d})^{\mu(I)}\leq (4d)^k$ on the number of leaves of the search tree.

Since finding a violating set can be done in time $\bigoh(m+n+\ell)$, computing an $(L,k)$-component or an irregular minimum $L$-$L'$ separator can be done in time $\bigoh(k^3(m+n))$ and there can be at most $k$ applications of Reduction Rule~\ref{red:regularity} along any root to leaf path of this search tree, we have the claimed bound of $\bigoh((4d)^{k}k^4(m+n+\ell))$ on the running time, completing the proof of Theorem~\ref{thm:skewsymmc}.\qed

\vspace{10 pt}

\noindent
We observe here that $\ell$ occurs in the running time simply because the time taken to compute a violating set is $\bigoh(\ell+m)$. However, in some cases, the set ${\cal T}$ may be given in the form of a violation oracle, in which case, the running time bound remains the same if violation oracle runs in time $\bigoh(\ell)$. Therefore, the theorem below follows from the proof of Theorem~\ref{thm:skewsymmc}.

\begin{theorem}\label{thm:skewsymmcimplicit}
  There is an algorithm for {\skewsymmc} that, given a tuple $(D=(V,A),\sigma,k)$ 
 along with a violation oracle for a family ${\cal T}$, runs in time $\bigoh((4d)^{k}k^4(\ell+m+n)$ and either returns a skew-symmetric multicut of size at most $2k$ or correctly concludes that no such set exists, where $\ell$ is the time required for the violation oracle to compute a violated set in the family ${\cal T}$, $m=\vert A\vert$ and $n=\vert V\vert$.
\end{theorem}

\section{Applications}\label{sec:apps}
In this section we use the algorithm developed for {\skewsymmc} to obtain linear time parameterized algorithms for several other problems.

\subsection{Linear time algorithm for {\almsatc}}
Algorithms for the {\almsatc} problem (defined in the introduction) together with its variable version have turned out to be extremely useful as a subroutine in several parameterized algorithms (see~\cite{MarxR09,MarxRazmulticut,LokshtanovR12}). The variable version of the problem is formally defined as follows.

\medskip

\defparproblem{{\almsat}}{A {\twocnf} formula $F$, integer $k$}{$k$}{Does there exist a set $S_v$ of at most $k$ variables such that by deleting all the clauses which they occur in from $F$ we get a satisfiable formula?}

\vspace{10 pt}


\noindent
It is  known that the \emph{variable} version {\almsat} can be reduced to the \emph{clause} version {\almsatc} via a linear time reductions~\cite{MarxRazmulticut}. Therefore, it suffices for us to give a reduction from the clause version of {\almsatc} to {\skewsymmc}. In order to give this reduction, we begin by recalling the notion of implication graphs of a {\twocnf} formula.

\begin{definition}\label{def:implgraph}
  Given a {\twocnf} formula $F$, the {\bf implication graph} of $F$ is
  denoted by \D{$F$} and is defined as follows. The vertex set of the
  graph is the set of literals of $F$ and for every clause $\{l_1,
  l_2\}$ in $F$, we have arcs $(\bar l_1,l_2)$ and $(\bar
  l_2,l_1)$.
\end{definition}

\noindent
Clearly, implication graphs are skew-symmetric where the involution $\sigma$ is defined as $\sigma(l)=\bar l$ and $\sigma(l_1,l_2)=(\bar l_2,\bar l_1)$.

\begin{theorem}{\sc \cite{AspvallPT79}}\label{thm:2sat testing} A {\twocnf} formula $F$ is satisfiable iff no literal and its complement are contained in the same strongly connected component of \D{$F$}.
\end{theorem}

\begin{observation}\label{obs:commute}
Given a {\twocnf} formula $F$, let $D$ be the implication graph of this formula and let $C$ be a subset of the clauses of $F$. Let $C_D$ be the corresponding set of arcs in the graph $D$. Then, the implication graph of $F-C$ is the same as the graph $D\setminus C_D$.
\end{observation}

\begin{lemma}Let $F$ be a {\twocnf} formula on $n$ variables $x_1,\dots, x_n$. Then, $(F,k)$ is a {\Yes} instance of {\almsatc} iff $($\D{$F$}$,{\cal T}=\{\{x_1\},\dots,\{x_n\}\},k)$ is a {\Yes} instance of {\oneskewsymmc}.
 \end{lemma}

\begin{proof} Suppose that $C$ is a set of clauses such that $\vert C\vert \leq k$ and $F-C$ is satisfiable and let $C_D$ be the corresponding set of arcs in $D$. Then, by Theorem~\ref{thm:2sat testing}, no literal of $F$ appears in the same strongly connected component as its complement in the implication graph of $F-C$. However, by Observation~\ref{obs:commute}, we have that $C_D$ is a set of arcs such that no vertex and its conjugate lie in the same strongly connected component of $D\setminus C_D$, which implies that $C_D$ is a solution for the instance of {\oneskewsymmc}.

Conversely, let $C_D$ be a self-conjugate set of arcs such that $\vert C_D\vert\leq 2k$ and no vertex in ${\cal T}$ lies in the same strongly connected component as its conjugate in $D\setminus C_D$. Let $C$ be the set of clauses of $F$ corresponding to $C_D$. Since $C_D$ is self-conjugate, we have that $\vert C\vert\leq k$. Then, by Observation ~\ref{obs:commute} and Theorem~\ref{thm:2sat testing}, the formula $F-C$ is satisfiable, which implies that $C$ is indeed a solution for the instance of {\almsatc}. This completes the proof of the lemma. 
\end{proof}

\noindent
Since the graph $D(F)$ can be constructed in time $\bigoh(\vert F\vert)$ and has $\bigoh(\vert F\vert)$ arcs, we have the following theorem.

\begin{theorem}\label{thm:2sat}
  There is an algorithm that, given an instance $(F,k)$ of
  {\almsatc} ({\almsat}), runs in time $\bigoh(4^{k}k^4\ell)$ and either returns an assignment satisfying all but at most $k$ clauses of $F$ or correctly concludes that no such assignment exists.

\end{theorem}

\noindent
Furthermore, there are known linear time parameter preserving reductions from {\sc Edge Bipartization} to {\oct} and from {\oct} to {\almsatc} 
(see \cite[Page Number -- 72 ]{wernicke2003algorithmic} and \cite{KhotR02}). The reduction from {\edgebip} to {\oct} increases the the number of edges and vertices in the graph by a factor of $\bigoh(k)$ and the reduction from {\oct} to {\almsatc} is both parameter as well as size preserving.
Therefore, have the following corollaries.

\begin{theorem}
 {\sc Edge Bipartization} and {\sc Odd Cycle Transversal}  can be solved in time $\bigoh(4^kk^5(m+n))$ and $\bigoh(4^kk^4(m+n))$ respectively where $m$ and $n$ are the number of edges and vertices in the input graph respectively. \end{theorem}


\subsection{Linear time algorithm for {\delqbackdoor}}

A CNF formula $F$
is $\qhorn$ if there is a \emph{certifying function} $\beta : \var(F) \cup \overline{\var(F)}
\rightarrow \{0,\frac{1}{2},1\}$ with $\beta(x)=1-\beta(\bar x)$ for
every $x \in \var(F)$ such that $\sum_{l \in C}\beta(l) \leq 1$ for
every clause $C$ of $F$. In this subsection, we prove Theorem~\ref{thm:delapproxqhorn}.  We begin by recalling the notion of a quadratic cover given by Boros et al. \cite{BorosHammerSun94}.

\begin{definition}
Given a CNF formula $F$, the {\bf quadratic cover} of $F$, is a Krom formula denoted by $F_2$ and is defined as follows. Let $x_1,\dots,x_n$ be the variables of $F$. For every clause $C$, we have $\vert C\vert-1$ new variables $y^C_1,\dots, y^C_{\vert C\vert -1}$. We order the literals in each clause according to their variables, that is, a literal of $x_i$ will occur before a literal of $x_j$ if $i<j$. Let $l^C_1,\dots, l^C_{\vert C\vert}$ be the literals of the clause $C$ in this order. The quadratic cover is defined as. 

\medskip

\hspace{-15 pt}$F_2=\hspace{3 pt} \bigcup_{C\in F}\bigcup_{1\leq i\leq \vert C\vert-1}\{\{l^C_i, y^C_i\}, \{\bar y^C_i, l^C_{i+1}\}\}\hspace{5 pt}$ 
$\cup \hspace{5 pt}\bigcup_{C\in F}\bigcup_{1\leq i\leq \vert C\vert -2}\{\{\bar y^C_i,y^C_{i+1}\}\}$.
\end{definition}

\begin{observation}
\label{simpleobs}
 If $F$ is a CNF formula of length $\ell$, then the number of arcs in \D{$F_2$} is $\bigoh(\ell)$.
\end{observation}

\noindent
We require the following characterization of $\qhorn$ formulas.

\begin{lemma}[\cite{BorosHammerSun94}]
\label{lem:qhorn testing} A CNF formula $F$ is $\qhorn$ iff no clause of $F$ has three literals $l_1$, $l_2$, $l_3$ such that each $l_i$ and $\bar l_i$ are  in the same strongly connected component of \D{$F_2$}.
\end{lemma}

Recall that a \emph{deletion $\CCC$\hy backdoor set} of $F$
is a set $B$ of variables such that $F \justminus B \in \CCC$.  These characterizations allow us to give a  linear time reduction from  {\delqbackdoor} to {\threeskewsymmc}.

\begin{theorem}\label{thm:delapproxqhorn}
  There is an algorithm that, given an instance $(F,k)$ of
  {\delqbackdoor}, runs in time $\bigoh(12^{k}k^5\ell)$ and either returns a deletion \qhorn\hy backdoor set of $F$ of size at
  most $k$ or correctly concludes that no such set exists, where $\ell$ is the length of $F$.
\end{theorem} 

\begin{proof}
The proof is by a reduction to {\threeskewsymmc}. We first construct the graph $D(F_2)$. We now define a graph $D_1$ which is a modification of $D(F_2)$ as follows. For every vertex $l_i$ in $D(F_2)$ corresponding to a positive literal in $F$, we have two vertices $l_i^+$ and $l_i^-$ and an arc $(l_i^-,l_i^+)$ and for every vertex $l_i$ in $D(F_2)$ corresponding to a negative literal we have two vertices $\bar l_i^+$ and $\bar l_i^-$ and an arc $(\bar l_i^+,\bar l_i^-)$. We say that an arc $(l_i^-,l_i^+)$ \emph{corresponds} to a (positive) literal $l_i$ and an arc $(\bar l_i^+,\bar l_i^-)$ \emph{corresponds} to a (negative) literal $l_i$.

Now, for every vertex $y$ in $D(F_2)$ which does not correspond to a literal of $F$, we add vertices $y_1,\dots, y_{k+1}$ and for every arc $(y,l_i)$ in $D(F_2)$, if $l_i$ is a positive literal, then we add arcs $(y_1,l_i^-),\dots,(y_{2k+1},l_i^-)$
and if $l_i$ is a negative literal, then we add arcs $(y_1,\bar l_i^-),\dots,(y_{2k+1},\bar l_i^-)$. For every arc $(l_i,y)$ in $D(F_2)$, if $l_i$ is a positive literal then we add arcs $(l_i^+,y_1),\dots,(l_i^+,y_{2k+1})$ and if $l_i$ is a negative literal then we add arcs $(\bar l_i^-,y_1),\dots,(\bar l_i^-,y_{2k+1})$. This completes the construction of $D_1$. Clearly, $D_1$ is also skew-symmetric. The purpose of modifying the graph $D(F_2)$ is simply to map literals of the input formula to arcs in the skew-symmetric graph and conversely to ensure that arcs which do not correspond to literals of the formula $F$ are unlikely to participate in skew-symmetric multicuts of size at most $k$.
We note that $\{l_1^+,l_2^+\}$ are contained in the same strongly connected component of $D_1$ if and only if $\{l_1,l_2\}$ are in the same strongly connected component of $D(F_2)$.

We now claim that $(F,k)$ is a {\Yes} instance of {\delqbackdoor} iff $(D_1,{\cal T},k,\emptyset)$ is a {\Yes} instance of {\threeskewsymmc} where ${\cal T}$ is the set of all triples of literals $\{l_1^+,l_2^+,l_3^+\}$ in $F$ such that $l_1,l_2,l_3$ occur in a clause in $F$.

Consider a solution $S$ for the instance of {\delqbackdoor} and let $S_D$ be the set of those arcs in $D_1$ which correspond to the literals of the variables in $S$. Clearly, $S_D$ is self-conjugate and $\vert S_D\vert\leq 2k$. We claim that $S_D$ is a skew-symmetric multicut for the given instance. If this were not the case, then there is a clause $C\in \mathscr{C}(F)$  and literals $l_1,l_2,l_3\in C$ such that $l_1^+,l_2^+,l_3^+$ each lie in the same strongly connected component of $D_1\setminus S_D$ as their complements.  Recall that by $\mathscr{C}(F)$ we denote the set of clauses of a CNF formula $F$.

However, by Lemma~\ref{lem:qhorn testing}, there is no violating triple in the graph $D(F_2) \setminus \lit(S)$ and therefore, there cannot be a violated set in the graph $D_1\setminus S_D$.

Conversely, consider a solution $S_D$ for the instance of {\threeskewsymmc}. It is easy to see from the construction of $D_1$ that $S_D$ is disjoint from arcs incident on any $y^C_i$. Therefore, the arcs in $S_D$ correspond to literals and hence variables in $F$. Let $S$ be the this set of variables. We claim that $S$ is a {\qhorn} deletion backdoor set. If this were not the case, then by Lemma~\ref{lem:qhorn testing}, there is a clause $C\in \mathscr{C}(F)$ and literals $l_1,l_2,l_3\in C$ such that $l_1,l_2,l_3$ each lie in the same strongly connected component of $D(F_2)\setminus \lit(S)$ as their complements. However, this implies that  $l_1^+,l_2^+,l_3^+$ each lie in the same strongly connected component of $D_1\setminus S_D$ as their complements, which is a contradiction. This completes the proof of correctness of the reduction.

Though this reduction is parameter preserving, it is not a linear time reduction since the number of triples we need to give as input to the {\threeskewsymmc} instance could be super linear in the length of the formula. However, by using an algorithm by Boros et al.~\cite{BorosHammerSun94} that runs in time $\bigoh(\ell)$ and returns a violated triple, we can use the above reduction which runs in time $\bigoh(k\ell)$ and returns a skew-symmetric graph with $\bigoh(k\ell)$ arcs, along with our algorithm for {\threeskewsymmc} to get an algorithm which runs in time $\bigoh(12^{k}k^5\ell)$. This completes the proof of the theorem.
\end{proof}

\noindent
Since every deletion $\qhorn$ backdoor set is also a strong $\qhorn$ backdoor set, Theorem~\ref{thm:delapproxqhorn} has the following corollary.

\begin{corollary}

 There is an algorithm for {\sc Satisfiability} that runs in time $\bigoh(12^kk^5\ell)$, where $k$ is the size of the smallest $\qhorn$ deletion backdoor set of the input formula.

\end{corollary}

 \section{Conclusions}
 
We introduced the {\skewsymmc} problem, a general graph separation problem which generalizes a large number of well-studied problems, and described an {\FPT} algorithm for this problem with a linear dependence on the input size and a moderate dependence on the parameter. This result gives the first linear time FPT algorithms for {\oct}, {\almsatc} and {\delqbackdoor}. We believe that there are more graph separation problems which can be reduced to {\skewsymmc} and that our algorithm can be used as a ``tool''  to give (linear time) {\FPT} algorithms for other problems which have graph separation at their core. We would like to remark that, to keep our analysis simple, we have not optimized the polynomial dependence of the running times on $k$.

We would also like to point out that the algorithms for variants of {\edgebip} and {\oct} studied in the paper of Marx et al.~\cite{MarxOSR13} use the almost linear time algorithm for {\oct} of Kawarabayashi and Reed or the quadratic time algorithm of Reed et al.~\cite{ReedSV04}. Therefore, using our algorithm instead of these algorithms results in linear time {\FPT} algorithms for these variants studied by Marx et al.

Finally, we leave open the kernelization complexity of this problem. Given that a (randomized) polynomial kernel for {\almsatc} exists~\cite{KratschW12}, it would be a natural goal to see if such a result extends to this much more general problem.

\longversion{
\bibliographystyle{siam}
\bibliography{literature}
}

\end{document}